\documentclass[a4paper,UKenglish,twoside,notitlepage]{article}
\usepackage{a4}
\usepackage[english]{babel}
\usepackage{color,graphicx}
\usepackage{amsmath,amssymb,amsthm,amsfonts}
\usepackage{mathrsfs}
\usepackage{verbatim}
\usepackage{bussproofs}
\usepackage{wasysym}
\usepackage{soul}
\usepackage{lscape}
\usepackage{cmll}
\usepackage{diagrams}
\usepackage{MnSymbol}
\usepackage{fullpage}
\usepackage{hyperref}

\newcommand{\vect}[1]{\boldsymbol{#1}}

\newcommand{\p}[0]{\mathcal{P}}

\newcommand{\lbi}[3]{\mathsf{let}\;#1\;\mathsf{be}\;#2\;\mathsf{in}\;#3}
\long\def\symbolfootnote[#1]#2{\begingroup%
\def\thefootnote{\fnsymbol{footnote}}\footnote[#1]{#2}\endgroup}
\newcommand{\exeq}[0]{\stackrel{!}{=}}
\newcommand{\ra}[1]{\stackrel{#1}{\longrightarrow}}

\newarrow{Squiggle}{dash}{dash}{dash}{dash}{O}
\newcommand{\sem}[2][M\!,g]{ [\![ #2 ]\!]^{}}

\newcommand{\ct}[1]{\underline{#1}}
\newcommand{\return}[0]{\mathsf{return}\;}
\newcommand{\sipm}[4]{\mathsf{pm}\;#1\;\mathsf{as}\;\langle #2,#3\rangle\;\mathsf{in}\;#4}
\newcommand{\upm}[2]{\mathsf{pm}\;#1\;\mathsf{as}\;\langle \rangle\;\mathsf{in}\;#2}
\newcommand{\force}[0]{\mathsf{force}\;}
\newcommand{\thunk}[0]{\mathsf{thunk}\;}
\newcommand{\toin}[3]{#1\;\mathsf{to}\;#2\;\mathsf{in}\;#3}

\newcommand{\Bcat}[0]{\mathcal{B}}
\newcommand{\Ccat}[0]{\mathcal{C}}
\newcommand{\Dcat}[0]{\mathcal{D}}
\newcommand{\proj}[2]{\mathbf{p}_{#1,#2}}
\newcommand{\diag}[2]{\mathsf{diag}_{#1,#2}}
\newcommand{\qu}[2]{\mathbf{q}_{#1,#2}}
\newcommand{\Set}[0]{\mathsf{Set}}
\newcommand{\diagv}[2]{\mathbf{v}_{#1,#2}}
\newcommand{\id}[0]{\mathsf{id}}
\newcommand{\Cat}[0]{\mathsf{Cat}}
\newcommand{\Id}[0]{\mathsf{Id}}
\newcommand{\self}[0]{\mathsf{self}}

\newcommand{\inl}[0]{\mathsf{inl}}
\newcommand{\inr}[0]{\mathsf{inr}}
\newcommand{\refl}[0]{\mathsf{refl}\;}
\newcommand{\type}[0]{\;\;\mathsf{type}}
\newcommand{\vtype}[0]{\;\;\mathsf{vtype}}
\newcommand{\ctype}[0]{\;\;\mathsf{ctype}}
\newcommand{\idpm}[3]{\mathsf{pm}\;#1\;\mathsf{as}\;(\refl #2)\; \mathsf{in}\;#3}
\newcommand{\Fam}[0]{\mathsf{Fam}}
\newcommand{\nil}[0]{\mathsf{nil\;}}
\newcommand{\tr}[0]{\mathsf{tr}\;}
\newcommand{\ctxt}[0]{\;\;\mathsf{ctxt}}
\newcommand{\print}[1]{\mathsf{print}\; #1\;.\; }
\newcommand{\diverge}[0]{\mathsf{diverge}\;}
\newcommand{\nondet}[2]{\mathsf{choose}_{#1}(#2)}
\newcommand{\readcell}[2]{\mathsf{readto}_{#1}(#2)}
\newcommand{\writecell}[1]{\mathsf{write}\;#1\;.\;}
\newcommand{\error}[1]{\mathsf{error\;#1\;}}

\theoremstyle{plain}
\newtheorem{theorem}{Theorem}[section]
\newtheorem{definition}[theorem]{Definition}

\newtheorem{remark}[theorem]{Remark}
\newtheorem*{claim*}{Claim}

\begin{document}
\title{\textbf{A Framework for  Dependent Types and Effects}}
\author{Matthijs V\'ak\'ar}
\date{Oxford, UK, \today}

\maketitle

\begin{abstract}We extend Levy's call-by-push-value (CBPV) analysis from simple to dependent type theory, to gain a better understanding of how to combine dependent types with effects. We define a dependently typed extension of CBPV, dCBPV-, and show that it has a very natural small-step operational semantics, which satisfies subject reduction and (depending on the effects present) determinism and strong normalization, and an elegant categorical semantics, which - surprisingly - is no more complicated than the simply typed semantics.

We have full and faithful translations from a dependently typed version of Moggi's monadic metalanguage and of a call-by-name (CBN) dependent type theory into dCBPV- which give rise to the expected operational behaviour. However, it turns out that dCBPV- does not suffice to encode call-by-value (CBV) dependent type theory or the strong (dependent) elimination rules for positive connectives in CBN-dependent type theory.

To mend this problem, we discuss a second, more expressive system dCBPV+, which additionally has a principle of Kleisli extension for dependent functions. We obtain the desired CBV- and CBN-translations of dependent type theory into dCBPV+. It too has a natural categorical semantics and operational semantics. However, depending on the effects we consider, we may lose uniqueness of typing, as the type of a computation may become more specified as certain effects are executed. This idea can be neatly formalized using a notion of subtyping.

This work arose as a generalisation from our previous work on linear dependent type theory to more general, possibly non-commutative effects. In this vein, we shall end with a brief discussion of a dependently typed version of the enriched effect calculus and various other extensions of dCBPV.

We hope that the theoretical framework of this paper on the one hand provides at least a partial answer to the fundamental type theoretic question of how one can understand the relationship between computational effects and dependent types. On the other hand, we hope it can contribute a small-step towards the ultimate goal of an elegant  fully fledged language for certified effectful programming.
\end{abstract}

\clearpage
\tableofcontents

\clearpage
\section{Introduction}
Dependent types \cite{martin1998intuitionistic,hofmann1997syntax} are slowly being taking up by the functional programming community and are in the transition from a quirky academic hobby to a practical approach to building certified software. Purely functional dependently typed languages like Coq \cite{Coq:manual} and Agda \cite{norell2007towards} have existed for a long time. If the technology is to become more widely used in practice, however, it is crucial that dependent types can be smoothly combined with the wide range of effects that programmers make use of in their day to day work, like non-termination and recursion, mutable state, input and output, non-determinism, probability and non-local control.

Although some languages exist which combine dependent types and effects, e.g. Cayenne \cite{augustsson1998cayenne}, $\Pi\Sigma$ \cite{altenkirch2010pisigma}, Zombie \cite{casinghino2014combining}, Idris \cite{brady2013idris}, Dependent ML \cite{xi1999dependent}, F* \cite{swamy2015dependent}, there have always been some strict limitations. For instance, the first four only combine dependent types with first class recursion (although Idris has good support for emulating other effects), Dependent ML constrains types to depend only on static natural numbers and F* does not allow types to depend on effectful terms at all (including non-termination). The sentiment of most papers discussing the marriage of these ideas seems to be that dependent types and effects form a difficult though not impossible combination. However, as far as we are aware, no general theoretical analysis has been given which discusses on a conceptual level the possibilities, difficulties and impossibilities of combining effects and dependent types.\\
\\
In a somewhat different vein, there has long been an interest in combining linearity and dependent types. This was first done in Cervesato and Pfenning's LLF \cite{cervesato1996linear}. Recently, a semantic analysis of LLF was given by the author in \cite{vakar2014syntax,vakar2015syntax} which has  proved important e.g. in the development of a game semantics for dependent type theory. One aspect that this abstract semantics as well as the study of particular models highlight is - more so than in the simply typed case - the added  insight and flexibility obtained by decomposing the $!$-comonad into an adjunction\footnote{Indeed, connectives seem to be most naturally formulated one either the linear or cartesian side: $\Sigma$- and $\Id$-constructors operate on cartesian types while $\Pi$-constructors operate on linear types.}. At the syntactic level this corresponds to working with dependently typed version of Benton's LNL-calculus \cite{benton1995mixed} rather than Barber and Plotkin's DILL \cite{barber1996dual}, as was done in \cite{krishnaswami2015integrating}.

Similarly, it has proved problematic to give a dependently typed version of Moggi's monadic metalanguage \cite{moggi1991notions}. We hope that this paper illustrates that also in this case a decomposed adjunction perspective, like that of call-by-push-value \cite{levy2012call}, is more flexible than a monadic perspective. Recall from \cite{benton1996linear} that if we decompose both linear logic and the monadic metalanguage into an adjunction, we can see the former to be a restricted case of the latter which only describes (certain) commutative effects.\\
\\
We show that the analysis of linear dependent type theory of \cite{vakar2015syntax,krishnaswami2015integrating,vakar2014syntax} generalises straightforwardly to general (non-commutative) effects to give a dependently typed call-by-push-value calculus that we call dCBPV-, which allows types to depend on values but which lacks a Kleisli extension (or sequencing) principle for dependent functions. This calculus is closely related to Harper and Licata's dependently typed polarized intuitionistic logic \cite{licata2009positively}. Its categorical semantics is obtained from that for linear dependent types, by relaxing a condition on the adjunction which would normally imply, among other things, the commutativity of the effects described. It straightforwardly generalises Levy's adjunction models for call-by-push-value \cite{levy2005adjunction} (from locally indexed categories to more general comprehension categories \cite{jacobs1993comprehension}) and, in a way, simplifies Moggi's strong monad models for the monadic metalanguage \cite{moggi1991notions}, as was already anticipated by Plotkin in the late 80s: in a dependently typed setting the monad strength follows straightforwardly from the natural demand that its adjunction is compatible with substitution and, similarly, the distributivity of coproducts follows from their compatibility with substitution. In fact, we believe call-by-push-value is most naturally formulated in a dependently typed setting, especially from a categorical semantic point of view. The small-step operational semantics for CBPV of \cite{levy2012call} transfers to dCBPV- without any difficulties.

When formulating the obvious candidate CBV- and CBN-translations of dependent type theory into dCBPV-, it becomes apparent that the CBN-translation is only well-defined if we work with the weak (non-dependent) elimination rules for positive connectives, while the CBV-translation is ill-defined altogether. To obtain a CBV-translation and the CBN-translation in all its glory, we have to add a principle of Kleisli extensions (or sequencing) for dependent functions to dCBPV-. We call the resulting calculus dCBPV+, to which we can easily extend our categorical and operational semantics. Normalization and determinacy results for the operational semantics remain the same. However, depending on the effects we consider, we may need to add extra coercion rules to the calculus to salvage subject reduction. These mechanisms of subtyping \cite{aspinall1996subtyping,luo1999coercive} embody the idea that a computation can obtain a more precise type after certain effects, like non-deterministic branching, are executed. In particular, they break uniqueness of typing. However, instead, we would expect a notion of minimal typing.

Before we conclude, we include a discussion of the possibility of adding some additional connectives and the matter of a dependently typed enriched effect calculus \cite{egger2009enriching}. Finally, we discuss the possible uses of the dCBPV- and dCBPV+ as a basis for effectful certified programming and indicate some of the obvious next steps in this line of work.\\
\\
We hope this analysis gives a helpful theoretical framework to study various combinations of dependent types and effects. In particular, it should give a robust motivation for the equations we should expect to hold in both CBV- and CBN-versions of effectful dependent type theory, through their translations into dependently typed CBPV. Moreover, it explains why combinations of dependent types and effects are slightly more straightforward in CBN than in CBV and highlights the need for a form of subtyping if one wants to combine certain effects with dependent types in a CBV-fashion.

Furthermore, we note that an expressive system like CBPV or a monadic language, where one has control over where effects occur, is even more important in presence of dependent types as it allows us to simultaneously use the language as an effectful programming language and a pure logic. We hope that dCBPV can be expanded in future to yield a language that can serve both as an effectful programming language and as a logic for reasoning about effectful programs. This would be similar to Hoare Type Theory \cite{nanevski2006polymorphism} in its aim, but different in its implementation.

This work again emphasizes the phenomenon that we can straightforwardly let types depend on values, but not necessarily computations, as type dependency on values is almost forced on us by both the syntax and categorical semantics of CBPV. Perhaps the analogy in CBPV is that functions, like dependent types, always take values as input. The only way to have a function take a computation as input is to input a thunk of a computation. It is no different with dependent types. This should be compared to the situation in linear logic where types are only allowed to depend on cartesian terms and not linear ones and the situation in polarised logic where types are only allowed to depend on terms of positive types and not negative types. Generally, type dependency on codata seems to be a subtle issue.

\subsubsection*{Acknowledgements}
I thank Tim Zakian, Alex Kavvos, Sam Staton and Mario Alvarez-Picallo for many interesting discussions and Paul Blain Levy for his patient explanations. I am grateful to Samson Abramsky for all his support. The author is supported by the EPSRC and the Clarendon Fund.

\clearpage
\section{Simply Typed Call-by-Push-Value, A Recap}
We briefly recall the syntax and categorical semantics of a variant of simply typed call-by-push-value  \cite{levy2012call}. All the material on the syntax of simple CBPV in section \ref{sec:scbpvsyn} is due to Levy. In section \ref{sec:scbpvsem}, we discuss a modified but equivalent (in the categorical sense) presentation of Levy's categorical semantics of simple CBPV that makes the transition to dependent types more natural, after which we give a few examples of models in section \ref{sec:simplmod}. Next, we briefly discuss the small-step operational semantics for CBPV in section \ref{sec:simplop}. Finally, in section \ref{sec:simpleff}, we briefly mention how one proceeds to add effects to the pure CBPV calculus we had discussed so far, which was after all the point of the whole endeavour. For this, we mostly take an operational point of view.

\subsection{Syntax} \label{sec:scbpvsyn}

We encourage the reader to look at the syntax of call-by-push-value (CBPV) in the following slightly simplified way: as providing an adjunction decomposition of the monadic metalanguage, similar (dual) to the one that Benton's LNL-calculus gives of (the comonadic) DILL, but in the more general setting of possibly non-commutative effects.

CBPV distinguishes between two classes of types: those of values and computations. These can be similarly read as, respectively, positive and negative types or as types of data and codata. The linear types of the LNL-calculus should be thought to arise as a special case of computation types, while its cartesian types live on the value side. The idea is that in natural deduction, for some connectives, the positive/value connectives, the introduction rule involves a choice, while the elimination rule is invertible (works through pattern matching) and for others, the negative/computation connectives, the opposite is true. As a rule of thumb, connectives that operate on computation types arise as left adjoint functors in the categorical semantics, while connectives that operate on value types are right adjoint functors.

Call-by-push-value (and polarised logic) chooses to keep the classes of types formed from both classes of connectives separate and adds two extra connectives $F$, which turns a value type into a type of computations that produce a result of the original value type, and $U$, which turns a computation type into a value type of thunks of computations of the original computation type. This allows us to use the full $\beta\eta$-equational theory for all connectives, even in the presence of effects. Importantly, we have call-by-value and call-by-name embeddings into CBPV, that give rise to their usual equational theories.\\
\\
CBPV has two classes of types - we sometimes underline types to emphasize that we mean a computation type:
\begin{align*}&\textnormal{value types \quad}A \qquad\qquad\qquad\textnormal{computation types \quad}\ct{B}.\end{align*}
In this paper, simple value and computation types are formed using the connectives of figure \ref{fig:vctypes}, excluding general inductive and coinductive types. Here, $1,\times$ will denote pattern-matching products, while $\top,\&$ are projection products\footnote{Note that these correspond to the two ways of defining products in the categorical semantics: as left adjoints to the internal hom or as right adjoints to the diagonal functor.}. More generally, following Levy, we include primitives $\Pi_{1\leq i\leq n}\ct{B_i}$ for $n$-ary projection products and $\Sigma_{1\leq i\leq n}A_i$ for $n$-ary sum (we write nullary and binary sum as $0$ and $A+A'$). We do this to emphasize their similarity to $\Pi$- and $\Sigma$-types in the dependently typed version of CBPV.

\begin{figure}[!ht]
\begin{tabular}{c|c}
\textbf{value/positive types} $A$& \textbf{computation/negative types} $\ct{B}$\\
\hline
$0$, $A + A'$, $\Sigma_{1\leq i\leq n}A_i$ & $A \Rightarrow \ct{B}$\\
$1$, $A \times A'$ & $\top$, $\ct{B} \& \ct{B'}$, $\Pi_{1\leq i\leq n}\ct{B_i}$\\
$U\ct{B}$ & $FA$\\
(inductive types) & (coinductive types)
\end{tabular}
\caption{\label{fig:vctypes} An overview of the simple value and computation types we consider with exception of general inductive and coinductive types which we shall not attempt to incorporate.}
\end{figure}

Similarly, CBPV has separate typing judgements for terms representing values and computations, respectively,
\begin{align*}
\Gamma\vdash^v a:A\qquad\qquad\qquad
\Gamma\vdash^c b:\ct{B}.
\end{align*}
Here, $\Gamma$ is a context, or list $x_1:A_1,\ldots,x_n:A_n$ of declarations of distinct identifiers $x_i$ of value type $A_i$. As usual, we distinguish between free and bound (i.e. non-free) identifiers and consider terms up to $\alpha$-equivalence, or permutation of their bound identifiers. The rules of the type theory force the free identifiers of a well-typed term to be declared in the context. For notational convenience, we treat indices $i$ of (terms of) a sum $\Sigma_{1\leq i\leq n}A_i$ or product $\Pi_{1\leq i\leq n}\ct{B_i}$ similarly to bound identifiers. A proper formal treatment would involve including the indices and their range in the context, to distinguish between bound and free indices and to consider freshness of the appropriate indices in various $\eta$-rules. We prefer to avoid this extra formality and keep their treatment informal as we are convinced the intended meaning will be clear to the reader and anyone so inclined can fill in the technical details. 

These typing judgements obey the rules of figure   \ref{fig:vcterms}.\vspace{-5pt}
\begin{figure}[!ht]
\centering
\fbox{\parbox{\linewidth}{
\begin{tabular}{ll}
\AxiomC{}
\UnaryInfC{$\Gamma,x:A,\Gamma'\vdash^v x:A$}
\DisplayProof\hspace{50pt} & 
\AxiomC{$\Gamma\vdash^v V:A$}
\AxiomC{$\Gamma,x:A,\Gamma'\vdash^{v/c} R:{B}$}
\BinaryInfC{$\Gamma,\Gamma'\vdash^{v/c} \lbi{x}{V}{R} :{B}$}
\DisplayProof\\
&\\
\AxiomC{$\Gamma\vdash^v V:A$}
\UnaryInfC{$\Gamma\vdash^c \return\; V:FA$}
\DisplayProof &
\AxiomC{$\Gamma\vdash^c M:FA$}
\AxiomC{$\Gamma,x:A,\Gamma'\vdash^c N:\ct{B}$}
\BinaryInfC{$\Gamma ,\Gamma'\vdash^c \toin{M}{x}{N}:\ct{B}$}
\DisplayProof\\
&\\
\AxiomC{$\Gamma\vdash^c M:\ct{B}$}
\UnaryInfC{$\Gamma\vdash^v \thunk M:U\ct{B}$}
\DisplayProof
&
\AxiomC{$\Gamma\vdash^v V: U\ct{B}$}
\UnaryInfC{$\Gamma\vdash^c \force V: \ct{B}$}
\DisplayProof\\
&\\
\AxiomC{$\Gamma\vdash^v V_i: A_i$}
\UnaryInfC{$\Gamma\vdash^v \langle i,V_i\rangle : \Sigma_{1\leq i\leq n}A_i$}
\DisplayProof
&
\AxiomC{$\Gamma\vdash^v V: \Sigma_{1\leq i\leq n}A_i$}
\AxiomC{$\{\Gamma,x:A_i\vdash^{v/c} R_i : {B}\}_{1\leq i\leq n}$}
\BinaryInfC{$\Gamma\vdash^{v/c} \sipm{V}{i}{x}{R_i} : {B}$}
\DisplayProof\\
&\\
\AxiomC{}
\UnaryInfC{$\Gamma\vdash^v\langle\rangle :1$}
\DisplayProof&
\AxiomC{$\Gamma\vdash^v V:1$}
\AxiomC{$\Gamma\vdash^{v/c} R:{B}$}
\BinaryInfC{$\Gamma\vdash^{v/c} \upm{V}{R}:{B}$}
\DisplayProof
\\
&\\
\AxiomC{$\Gamma\vdash^v V_1:A_1$}
\AxiomC{$\Gamma\vdash^v V_2:A_2$}
\BinaryInfC{$\Gamma\vdash^v \langle V_1,V_2\rangle :A_1\times A_2$}
\DisplayProof
&
\AxiomC{$\Gamma\vdash^v V: A_1\times A_2$}
\AxiomC{$\Gamma,x:A_1,y:A_2\vdash^{v/c} R:{B}$}
\BinaryInfC{$\Gamma\vdash^{v/c} \sipm{V}{x}{y}{R}:{B}$}
\DisplayProof\\
&\\
\AxiomC{$\{\Gamma\vdash^c M_i :\ct{B_i}\}_{1\leq i\leq n}$}
\UnaryInfC{$\Gamma\vdash^c \lambda_i M_i : \Pi_{1\leq i\leq n}\ct{B_i}$}
\DisplayProof
&
\AxiomC{$\Gamma\vdash^c M: \Pi_{1\leq i\leq n}\ct{B_i}$}
\UnaryInfC{$\Gamma\vdash^c i\textquoteleft M : \ct{B_i}$}
\DisplayProof\\
&\\
\AxiomC{$\Gamma,x:A\vdash^c M:\ct{B}$}
\UnaryInfC{$\Gamma\vdash^c \lambda_xM:A\Rightarrow\ct{B}$}
\DisplayProof
&
\AxiomC{$\Gamma\vdash^v V:A$}
\AxiomC{$\Gamma\vdash^c M:A\Rightarrow\ct{B}$}
\BinaryInfC{$\Gamma\vdash^c V\textquoteleft M : \ct{B}$}
\DisplayProof
\end{tabular}
}
}
\caption{\label{fig:vcterms} Values and computations of simple CBPV. A rule involving $\vdash^{v/c}$ should be interpreted as a shorthand for two rules: one with $\vdash^v$ and one with $\vdash^c$ in both the hypothesis and conclusion.}
\end{figure}\quad\\
We consider these terms up to $\alpha$-equivalence and define an operational semantics of them as such in section \ref{sec:simplop}. We frequently also consider the terms up to the equational theory of figure \ref{fig:vceqs} together with the rules which state that all term formers respect equality and that equality is an equivalence relation, where we write $M[V/x]$ for the syntactic meta-operation of capture avoiding substitution of $V$ in $M$. We shall see that this equational theory naturally arises from the categorical semantics of simple CBPV.
\begin{figure}[!ht]
\fbox{
\parbox{\linewidth}
{
\begin{tabular}{ll}
$\lbi{x}{V}{R}=R[V/x]$ &\\
$\toin{(\return V)}{x}{M} = M[V/x]$ & $M =\toin{M}{x}{\return x}$\\
$\force\thunk M=M$ & $V=\thunk\force V$\\
$\sipm{\langle i, V\rangle }{i}{x}{R_i}=R_i[V/x]$ & $R[V/z]\stackrel{\#x}{=}\sipm{V}{i}{x}{R[\langle i,x\rangle/z]}$\\
$\upm{\langle\rangle}{R}=R$ & $R[V/z]=\upm{V}{R[\langle\rangle/z]}$\\
$\sipm{\langle V,V'\rangle }{x}{y}{R}=R[V/x,V'/y]$ & $R[V/z]\stackrel{\#x,y}{=}\sipm{V}{x}{y}{R[\langle x,y\rangle/z]}$\\
$i\textquoteleft \lambda_j M_j=M_i$ & $M=\lambda_i i\textquoteleft M$\\
$V\textquoteleft \lambda_x M=M[V/x]$ & $M\stackrel{\#x}{=}\lambda_x x\textquoteleft M$\\
$\toin{(\toin{M}{x}{N} )}{y}{N'}\stackrel{\#x}{=}\toin{M}{x}{(\toin{N}{y}{N'})}$ &\\
$\toin{M}{x}{\lambda_i N_i}=\lambda_i (\toin{M}{x}{N_i})$ & \\
$\toin{M}{x}{\lambda_y N}\stackrel{\#y}{=}\lambda_y(\toin{M}{x}{N})$ & \\
\end{tabular}
}
}
\caption{\label{fig:vceqs} Equations of simple CBPV. These should be read as equations of typed terms: we impose them if we can derive that both sides of the equation are well-typed terms of the same type in the same context. We write $\stackrel{\#x_1,\ldots,x_n}{=}$ to indicate that for the equation to hold, the identifiers $x_1,\ldots,x_n$ should, in both terms being equated, be replaced by fresh identifiers, in order to avoid unwanted variable bindings.}
\end{figure}\quad\\
\\
Recall that a call-by-value (CBV) and call-by-name (CBN) evaluation strategy on the $\lambda$-calculus generally give rise to different equational theories (in presence of effects) \cite{plotkin1975call}. For instance, the $\eta$-rule for function types typically fails in the former and that for sum types in the latter. CBPV gives rise to both of these equational theories by embedding the $\lambda$-calculus either with a CBV- or with a CBN-translation. 

In presence of effects, the usual pure connectives of products, coproducts and function types bifurcate into many variants due to the distinction of versions of different arities and the distinction between projection and pattern matching products. These are nicely and uniformly treated in Levy's Jumbo $\lambda$-calculus \cite{levy2006jumbo}. There are fully faithful translations $(-)^v$ and $(-)^n$, respectively, from CBV- and CBN-versions of this whole calculus into CBPV \cite{levy2012call} and, in fact, the same is true if we consider arbitrary theories rather than the pure calculi. To convey the intuition without getting stuck on technicalities, we present some special cases of the translations in figures \ref{fig:cbvtrans} and \ref{fig:cbntrans}.

\begin{figure}[!ht]
\fbox{
\parbox{\linewidth}{
\begin{tabular}{l|l||l|l}
\textbf{CBV type}  & \textbf{CBPV type} & \textbf{CBV term } & \textbf{CBPV term}\\
\hline
$A$ & $A^v$ & $x_1:A_1,\ldots,x_m:A_m\vdash M:A$ & $x_1:A_1^v,\ldots,x_m:A_m^v\vdash^c M^v:F(A^v)$\\
 && $x$& $\return x$\\
  &&$\lbi{x}{M}{N}$ & $\toin{M^v}{x}{N^v}$ \\
$\Sigma_{1\leq i\leq n }A_i$ & $\Sigma_{1\leq i\leq n }A_i^v$& $\langle i,M\rangle $&$\toin{M^v}{x}{\return \langle i, x\rangle }$ \\
 &&$\sipm{M}{i}{x}{N_i}$& $\toin{M^v}{z}{(\sipm{z}{i}{x}{N_i^v})}$\\
$\Pi_{1\leq i\leq n}A_i $ & $U\Pi_{1\leq i \leq n} FA_i^v$ &$\lambda_iM_i$ &$\return \thunk (\lambda_i  M_i^v)$\\
&&$i\textquoteleft  N $&$\toin{N^v}{z}{(i\textquoteleft \force z)}$\\
$A\Rightarrow A'$ & $U(A^v \Rightarrow F A'^v)$ & $\lambda_x M$&$\return \thunk \lambda_x M^v$\\
&&$M\textquoteleft N$ &$\toin{M^v}{x}{(\toin{N^v}{z}{(x\textquoteleft \force z)})}$\\
$1$ & $1$ & $\langle\rangle$ & $ \return \langle\rangle$  \\
&&$\upm{M}{N}$&$\toin{M^v}{z}{(\upm{z}{N^v})}$\\
$A \times A'$ & $A^v \times A'^v$ & $ \langle M, N\rangle $  & $\toin{M^v}{x}{(\toin{N^v}{y}{\return \langle x,y\rangle})}$\\
&&$\sipm{M}{x}{y}{N}$&$\toin{M^v}{z}{(\sipm{z}{x}{y}{N^v})}$\\
\end{tabular}
}
}
\caption{\label{fig:cbvtrans} A CBV-translation of a simple $\lambda$-calculus into CBPV.}
\end{figure}

\begin{figure}[!ht]
\fbox{
\parbox{\linewidth}{
\begin{tabular}{l|l||l|l}
\textbf{CBN type}  & \textbf{CBPV type} & \textbf{CBN term } & \textbf{CBPV term }\\
\hline
$\ct{B}$ &  $\ct{B}^n$& $x_1:\ct{B}_1,\ldots,x_m:\ct{B}_m\vdash M:\ct{B}$& $x_1:U\ct{B}_1^n,\ldots,x_m:U\ct{B}_m^n\vdash^c M^n:\ct{B}^n$ \\
 && $x$& $\force x$\\
  & & $\lbi{x}{M}{N}$ & $\lbi{x}{(\thunk M^n)}{N^n}$ \\
$\Sigma_{1\leq i\leq n }\ct{B}_i$ & $F\Sigma_{1\leq i\leq n }U\ct{B}_i^n$& $\langle i,M\rangle $&$\return \langle i,\thunk M^n\rangle $ \\
 & &$\sipm{M}{i}{x}{N_i}$&$\toin{M^n}{z}{(\sipm{z}{i}{x}{N_i^n})}$ \\
$\Pi_{1\leq i\leq n}\ct{B}_i $ & $\Pi_{1\leq i \leq n} \ct{B}_i^n$ & $\lambda_iM_i$& $\lambda_iM_i^n$\\
 && $i\textquoteleft M$ & $i\textquoteleft M^n$\\
$\ct{B}\Rightarrow \ct{B'}$ & $(U\ct{B}^n)\Rightarrow \ct{B'}^n$ & $\lambda_x M $& $\lambda_xM^n$\\
 &&$N\textquoteleft M$ & $(\thunk N^n) \textquoteleft M^n$ \\
$1$ & $F1$ & $\langle\rangle$ & $\return \langle\rangle$  \\
&&$\upm{M}{N}$&$\toin{M^n}{z}{(\upm{z}{N^n})}$\\
$\ct{B} \times \ct{B'}$ & $F(U\ct{B}^n \times U\ct{B'}^n)$ & $\langle M, N\rangle $  & $\return \langle \thunk M^n,\thunk N^n\rangle$\\
&& $\sipm{M}{x}{y}{N}$& $\toin{M^n}{z}{(\sipm{z}{x}{y}{N^n})}$ \\
\end{tabular}
}
}
\caption{\label{fig:cbntrans} A CBN-translation of a simple $\lambda$-calculus into CBPV.}
\end{figure}

\clearpage
\subsection{Categorical Semantics}
\label{sec:scbpvsem}
CBPV admits a simple notion of a categorical model. We present a variation of that of \cite{levy2005adjunction} to allow a smooth transition to dependent types.

\subsubsection{Intermezzo: Categorical Semantics of Pure Dependent Type Theory}
The reader might wonder why we pause to recall the semantics of pure dependent type theory, before discussing the semantics of simple CBPV. The trick is that the latter is treated most easily in terms of concepts developed for the former. This reflects one of the points this paper is trying to make: paradoxically, the semantics of dependently typed CBPV is easier to describe - in a way - than the simply typed version.

\begin{definition}[Comprehension Axiom] Let $\Bcat^{op}\ra{\Ccat}\Cat$ be a strict\footnote{For brevity, from now on we shall drop the modifier "strict". For instance, if we mention an indexed honey badger, we shall really mean a strict indexed honey badger.} indexed category (writing $\Cat$ for  the category of small categories and functors). Given $B'\ra{f} B$ in $\Bcat$, let us write $-\{f\}$ for the change of base functor $\Ccat(B)\ra{}\Ccat(B')$. Recall that $\Ccat$ is said to satisfy the \emph{comprehension axiom} if
\begin{itemize}
\item $\Bcat$ has a terminal object $\cdot$;
\item all fibres $\Ccat(B)$ have terminal objects $1_B$ which are stable under change of base (for which we just write $1$);
\item the presheaves (writing $\Set$ for  the category of small sets and functions)
\begin{diagram}
(\Bcat/B)^{op} & \rTo & \Set\\
(B'\ra{f}B) & \rMapsto  & \Ccat(B')(1,C\{f\})
 \end{diagram}
are representable. That is, we have representing objects $B.C\ra{\proj{B}{C}}B$ and natural bijections
\begin{diagram}
\Ccat(B')(1,C\{f\})& \rTo^{\cong} &\Bcat/B(f,\proj{B}{C})\\
c & \rMapsto & \langle f,c\rangle .
\end{diagram}
\end{itemize}
We write $\diagv{B}{C}$ for the element of $\Ccat(B.C)(1,C\{\proj{B}{C}\})$ corresponding to $\id_{\proj{B}{C}}$ (the universal elements of the representation). We define the morphisms \mbox{\begin{diagram} B.C & \rTo^{\diag{B}{C}:=\langle\id_{B.C},\diagv{B}{C} \rangle} & B.C.C\{\proj{B}{C}\}\end{diagram}} and \mbox{\begin{diagram}B'.C\{f\} & \rTo^{\qu{f}{C}:=\langle\proj{B'}{C\{f\}};f,\diagv{B'}{C\{f\}}\rangle} & B.C\end{diagram}}. We have maps
\begin{diagram}
\Ccat(B)(C',C)& \rTo^{\proj{B}{-}} &\Bcat/B(\proj{B}{C'},\proj{B}{C})\\
c & \rMapsto & \langle \proj{B}{C'},\diagv{B}{C'};c\{\proj{B}{C'}\}\rangle .
\end{diagram}
When these are full and faithful, we call the comprehension \emph{full and faithful}, respectively. When it induces an equivalence $\Ccat(\cdot)\cong \Bcat$, we call the comprehension \emph{democratic}\footnote{This corresponds to the syntactic condition that all contexts are formed from the empty context by adjoining types.}.
\end{definition}
\begin{remark}
The definition of an indexed category with comprehension is easily seen to be equivalent to Jacobs' notion of a split comprehension category with unit. We prefer this formulation in terms of indexed categories as strictness is important in computer science, in which case the fibrational perspective is needlessly abstract. Jacobs' notion of fullness of a comprehension category corresponds - paradoxically - to our demand of the comprehension both being full and faithful. We believe it is useful to use this more fine-grained terminology.
\end{remark}
Recall that these categories are a standard notion of  model of pure dependent type theory~\cite{hofmann1997syntax}.
\begin{theorem}[Pure DTT Semantics] 
We have a sound interpretation of pure dependent type theory with $1$-types in any  indexed category $\Bcat^{op}\ra{\Ccat}\Cat$ with full and faithful comprehension. We list necessary and sufficient conditions for the model to support various type formers\footnote{We have chosen to model the strong (or dependent) elimination rules for $0,+,\Sigma,\Id$-types, as their expressive power is important in practical dependently typed programming. We can also model the less expressive weak (or non-dependent) elimination rules in this framework as the existence of certain left adjoint functors that satisfy Beck-Chevalley conditions. Similarly, we could easily generalise this to intensional $\Id$-types. The reader can find an extensive treatment in \cite{jacobs1999categorical,hofmann1997syntax}.}:
\begin{itemize}
\item $0,+$-types - finite indexed coproducts (i.e. finite coproducts in all fibres that are stable under change of base) such that the following canonical morphisms are bijections
$$\Ccat(C.\Sigma_{1\leq i\leq n}C_i)(C',C'')\ra{}\Pi_{1\leq i\leq n}\Ccat(C.C_i)(C'\{\proj{C}{\langle i,\id_{C_i}\rangle }\},C''\{\proj{C}{\langle i,\id_{C_i}\rangle }\});
$$
\item $\Pi$-types - right adjoint functors $ -\{\proj{B}{C}\}\dashv \Pi_C$ satisfying the right Beck-Chevalley condition\footnote{This is a technical condition which corresponds to compatibility between $\Pi$-type formers and substitution, see \cite{jacobs1999categorical,vakar2014syntax}.};
\item $\Sigma$-types - objects $\Sigma_CD$ of $\Ccat(B)$ such that $\proj{B}{\Sigma_CD}=\proj{B.C}{D};\proj{B}{C}$;
\item (extensional) $\Id$-types - objects $\Id_C$ of $\Ccat(B.C.C)$ such that $\proj{B.C.C}{\Id_C}=\diag{B}{C}$.
\end{itemize}
In fact, the interpretation in such categories is complete in the sense that an equality holds in all interpretations iff it is provable in the syntax of dependent type theory.  
\end{theorem}

In particular, we can use such categories to model pure simple type theory with $0,+,1,\times,\Rightarrow$-types as a special case, rather than using the usual notion of model of a bicartesian closed category $\Ccat$. Indeed, starting from such a bicartesian closed category $\Ccat$, we can produce an indexed category $\Ccat^{op}\ra{\self(\Ccat)}\Cat$ where $\self(\Ccat)(A)$ has the same objects as $\Ccat$ and $\self(\Ccat)(A)(B,C)=\Ccat(A\times B,C)$ with the obvious identities and composition. With the change of base functors defined to be the identity on objects and to act on morphisms in the obvious way through precomposition. 

\begin{theorem}
For a bicartesian closed category $\Ccat$, $\self(\Ccat)$ is an indexed category with full and faithful democratic comprehension, which supports $0$-, $+$-, $\Sigma$- and $\Pi$-types. In this case, $\proj{A}{B}$ is the usual product projection from $A\times B\ra{}A$. It does not usually support extensional $\Id$-types as these correspond to objects $1/A$ such that $1/A\times A\cong 1$.
\end{theorem}

\subsubsection{Categorical Semantics of Simple CBPV}
We can now state what a categorical model of CBPV is. The philosophy is to add an extra (locally) indexed category $\Dcat$ to model computations (and, more generally, stacks\footnote{For elegance of presentation, we somewhat naughtily identify computations with certain stacks, rather than keeping them separate but isomorphic. One could also easily write down an equivalent (in the categorical sense) treatment \`a la Levy where the computations and stacks are really kept separate.}) separately from values and to demand all appropriate negative (right adjoint) connectives in $\Dcat$ and all positive (left adjoint) ones in $\self(\Ccat)$. The idea will be that values $\Gamma\vdash^v V:A$ denote elements of $\self(\Ccat)(\sem{\Gamma})(1,\sem{A})$ and that computations $\Gamma\vdash^c M:\ct{B}$  denote elements of $\Dcat(\sem{\Gamma})(F1,\sem{\ct{B}})$. We encourage the reader to think of the more general homsets $\Dcat(\sem{\Gamma})(\sem{B'},\sem{B})$ as the denotations of general stacks, which we treat in section \ref{sec:simplop}.

\begin{definition}[Simple CBPV Model] By a categorical model of simple CBPV, we shall mean the following data.
\begin{itemize}
\item A cartesian category $(\Ccat,1,\times)$ of \emph{values};\\
\item an indexed category $\Ccat^{op}\ra{\Dcat}\Cat$ of \emph{stacks} (and, in particular, \emph{computations}) where the change of base functors are a bijection on objects;\\
\item $0,+$-types in $\self(\Ccat)$\footnote{This amounts to having distributive finite coproducts in $\Ccat$.} such that, additionally, the following induced maps are bijections: $$\Dcat(C.\Sigma_{1\leq i\leq n} C_i )(\ct{D},\ct{D'})\ra{}\Pi_{1\leq i \leq n}\Dcat(C.C_i)(\ct{D},\ct{D'});$$
\item an indexed adjunction\footnote{As Plotkin pointed out at the time of Moggi's original work on the monadic metalanguage, this gives a strong monad $T=UF$ on $\Ccat$ \cite{moggi1988computational}.}\mbox{
\begin{diagram}
\Dcat & \pile{\lTo^F\\\bot\\\rTo_U} & \self(\Ccat);
\end{diagram}}
\item $\Pi$-types in $\Dcat$ in the sense of having right adjoint functors $-\{\proj{A}{B}\}\dashv \Pi_B:\Dcat(A)\ra{}\Dcat(A.B)$ satisfying the right Beck-Chevalley condition;
\item Finite indexed products $(\top,\&)$ in $\Dcat$;
\end{itemize}
Note that $\self(\Ccat)$ automatically has indexed terminal objects and $\Sigma$-types.
\end{definition}

\begin{theorem}[Simple CBPV Semantics] We have a sound interpretation of CBPV in a CBPV model:
\begin{align*}
\sem{\cdot} & = \cdot\\
\sem{\Gamma,x:A} &= \sem{\Gamma}.\sem{A}\\
\sem{\Gamma\vdash^v A}&=\self(\Ccat)(\sem{\Gamma})(1,\sem{A})\\
\sem{\Gamma\vdash^c \ct{B}} & = \Dcat(\sem{\Gamma})(F1,\sem{\ct{B}})\\
\sem{\Gamma;\ct{B}\vdash^k\ct{C}} & = \Dcat(\sem{\Gamma})(\sem{\ct{B}},\sem{\ct{C}})\\
\sem{1}&=1\\
\sem{A\times A'} &= \sem{A}\times \sem{A'}\cong\Sigma_{\sem{A}}\sem{A'}\{\proj{\sem{\Gamma}}{\sem{A}}\}\\
\sem{\Sigma_{1\leq i\leq n}A_i}&=(\cdots(\sem{A_1}+\sem{A_2})+\cdots)+\sem{A_n})\\
\sem{\Pi_{1\leq i\leq n}\ct{B}_i}&=(\cdots(\sem{\ct{B}_1}\&\sem{\ct{B}_2})\&\cdots)\&\sem{\ct{B}_n})\\
\sem{A\Rightarrow \ct{B}} & = \Pi_{\sem{A}}\sem{\ct{B}}\{\proj{\sem{\Gamma}}{\sem{A}}\}\\
\sem{FA}&=F\sem{A}\\
\sem{U\ct{B}} &= U\sem{\ct{B}},
\end{align*}
where we write $\Gamma;\ct{B}\vdash^k \ct{C}$ for the set of stacks of type $\ct{C}$ in context $\Gamma;\ct{B}$ (see section \ref{sec:simplop}), together with the obvious interpretation of terms. The interpretation in such categories is complete in the sense that an equality of values or computations holds in all interpretations iff it is provable in the syntax of CBPV. In fact, if we add the obvious admissible weakening and exchange rules to CBPV, we have an onto relationship between models and theories which satisfy mutual soundness and completeness results.
\end{theorem}
\begin{remark}
Note that these models contain quite a bit of redundancy and that completeness currently does not hold for stacks. The idea is that one could add complex stacks to the syntax of CBPV to get a precise 1-1 correspondence between syntactic theories in CBPV and categorical models  \cite{levy2012call}.
\end{remark} 
Let us write $T$ for the indexed monad $UF$ on $\self(\Ccat)$ and $!$ for the indexed comonad $FU$ on $\Dcat$. We note that the translations from CBV and CBN into CBPV correspond to interpreting CBV and CBN in the Kleisli and coKleisli categories for $T$ and $!$ respectively. More generally, we can note that the translations 
of figures \ref{fig:cbvtrans} and \ref{fig:cbntrans} can be transformed into semantic translations which means that any CBPV model gives rise to models of the CBV- and CBN-$\lambda$-calculus.

\begin{theorem}[Simple CBV-Semantics] We obtain a sound interpretation of the CBV-$\lambda$-calculus with $1,\times,\Rightarrow,\Sigma_{1\leq i\leq n},\Pi_{1\leq i\leq n}$-types in the Kleisli category for $T$:
\begin{align*}
\sem{A_1,\cdots,A_n\vdash A}&=\Dcat(\sem{A_1}. \cdots.\sem{A_n})(F1,F\sem{A})\cong\self(\Ccat)_T(\sem{A_1}. \cdots.\sem{A_n})(1,\sem{A}).
\end{align*}
The interpretation is complete with respect to this class of models.
\end{theorem}

\begin{theorem}[Simple CBN-Semantics] We obtain a sound interpretation of the CBN-$\lambda$-calculus with $1,\times,\Rightarrow,\Sigma_{1\leq i\leq n},\Pi_{1\leq i\leq n}$-types in the coKleisli category for $!$:
\begin{align*}
\sem{\ct{B_1},\cdots,\ct{B_n}\vdash \ct{B}}&=\Dcat(U\sem{\ct{B_1}}.\cdots.U\sem{\ct{B_n}})(F1,\sem{\ct{B}})\cong\Dcat_{!}(U\sem{\ct{B_1}}.\cdots.U\sem{\ct{B_n}})(\top,\sem{\ct{B}}).
\end{align*}
The interpretation is complete with respect to this class of models.
\end{theorem}
Again, both of these results could be strengthened to the statement that we have an onto relationship between models and theories which satisfy mutual soundness and completeness results.
\clearpage
\subsection{A Few Words about Models}\label{sec:simplmod}
An extensive discussion of particular models as well as comparisons between CBPV models and other notion of categorical models of effects can be found in \cite{levy2012call}. Here, we shall be very brief and just recall the following two results and provide some context for the relationship between effects and linear logic.

\begin{theorem}
Let $\Ccat\rightleftarrows \Dcat'$ be a Benton linear-non-linear adjunction model of intuitionistic exponential additive multiplicative linear logic \cite{benton1995mixed,mellies2009categorical} in the sense of an adjunction with strong monoidal left adjoint
\begin{diagram}
\Ccat & \pile{\rTo^F\\\bot\\\lTo_U} & \Dcat'
\end{diagram}
between a cartesian monoidal category $\Ccat$ and a symmetric monoidal closed category $\Dcat'$ with finite distributive coproducts in $\Ccat$ and finite products in $\Dcat'$. In that case, $\Ccat\rightleftarrows\Dcat'$ gives rise to a canonical model of simple CBPV where $UF$ is a commutative monad \cite{benton1996linear}.
\end{theorem}
\begin{proof}[Proof (sketch)] We define the indexed category $\Ccat^{op}\ra{\Dcat}\Cat$ as having the same objects as $\Dcat'$ in each fibre and morphisms $\Dcat(A)(B,C):=\Dcat'(B,FA\multimap  C)$.

\end{proof}
In this way, we can see that linear logic, in a way, only describes certain commutative effects. In fact, the following partial converse result was obtained by \cite{keigher1978symmetric}.
\begin{theorem}\label{thm:commtolinear}Let $\Ccat$ be a bicartesian closed category with a commutative monad $T$ and which additionally has equalizers and coequalizers. Then, the Eilenberg-Moore category $\Ccat^T$ is symmetric monoidal closed and the Eilenberg-Moore adjunction defines a linear-non-linear adjunction.
\end{theorem}
We can elaborate the above a bit by noting the following, using a neat little argument that I learnt from Sam Staton.
\begin{theorem}
Every cartesian closed category $\Ccat$ with a commutative monad $T$ embeds fully faithfully into a model of intuitionistic exponential additive multiplicative linear logic inducing the monad $T$.
\end{theorem}
\begin{proof}[Proof (sketch)]Let $\widehat{\Ccat}$ be the category of presheaves on $\Ccat$ (its cocompletion). We note that the Yoneda embedding defines a strict 2-functor from the 2-category of categories to the 2-category of cocomplete categories (computing its action on morphisms by taking Yoneda extensions). In fact, using the Day convolution \cite{day1970closed}, it defines a strict 2-functor from the 2-category of symmetric monoidal categories (with lax symmetric monoidal functors and monoidal natural transformations) to the (sub-) 2-category of cocomplete symmetric monoidal categories. Recalling that a commutative monad is precisely a monad in the 2-category of symmetric monoidal categories and lax symmetric monoidal functors and that 2-functors preserve monads, we get a commutative monad $\widehat{T}$ on $\widehat{\Ccat}$ which restricts to $T$ on $\Ccat$. Noting that $\widehat{\Ccat}$ is bicartesian closed with equalizers and coequalizers (a topos even), we can apply theorem \ref{thm:commtolinear} for the result that $\widehat{\Ccat}^{\widehat{T}}\leftrightarrows \widehat{\Ccat}$ defines a model of intuitionistic exponential additive multiplicative linear logic.\end{proof}

\begin{remark}We see that intuitionistic linear logic almost precisely describes all commutative effects, viewed from a CBN-point of view in case of dual intuitionistic linear logic \cite{barber1996dual} or from a CBPV point of view in case of linear-non-linear logic \cite{benton1995mixed}. Still, this point of view does not seem to be widely held. Perhaps this is due to the fact that in the (initial) syntactic model of linear logic, a so-called principle of uniformity of threads (called such because it implies the usual principle $\Dcat(!A,B)\cong \Dcat(!A,!B)$) holds: the unit of the adjunction $F\dashv U$ is an isomorphism $\id_\Ccat\cong UF=:T$. In this sense, the free linear logic model does not describe any effects from the monadic point of view. All its interesting information is contained in the comonad $!:=FU$ of the adjunction.
\end{remark}

CBPV models for possibly non-commutative effects can be obtained from any monad model \cite{moggi1991notions} of the monadic metalanguage \cite{levy2012call}.
\begin{theorem}Any bicartesian closed category $\Ccat$ with a strong monad $T$ gives rise to a  CBPV model.
\end{theorem}
\begin{proof}[Proof (sketch)] Recall that in this setting the Eilenberg-Moore category $\Ccat^T$ is cartesian closed. We define the indexed category $\Ccat^{op}\ra{\Dcat}\Cat$ to have the same objects as the Eilenberg-Moore category $\Ccat^T$ in each fibre and morphisms $\Dcat(A)(k,l):=\Ccat^T( k,\mu_A\Rightarrow l)$, where we write $\mu$ for the multiplication of $T$. $F\dashv U$ is interpreted by the usual Eilenberg-Moore adjunction.
\end{proof}

\clearpage 

\subsection{Operational Semantics}\label{sec:simplop}
Importantly, CBPV admits a natural operational semantics that, for ground terms, reproduces the usual operational semantics of CBV and CBN under the specified translations into CBPV \cite{levy2012call} and that can easily be extended to incorporate various effects that we may choose to add to pure CBPV. We only very briefly discuss this for the convenience of the reader. All of the material in this section is due to Levy \cite{levy2012call}.\\
\\
First, we restrict the terms of our syntax a bit, by excluding so-called \emph{complex values} which unnecessarily complicate the presentation of the operational semantics. Complex values are defined to be values containing $\mathsf{pm}\;\;\mathsf{as}\;\;\mathsf{in}\;\;$- and $\lbi{}{}{}$-constructs. As the normalization of  values do not produce effects, all reasonable evaluation strategies for them are observationally indistinguishable and we could choose our favourite.

Luckily it turns out that excluding complex values is not a restriction at all, at least from the extensional point of view, as we have the following result \cite{levy2006call}.
\begin{theorem}[Redundancy of Complex Values for Computations] For any CBPV computation $\Gamma\vdash^c M:\ct{B}$, there is a computation $\Gamma\vdash^c \widetilde{M}:\ct{B}$ which does not contain complex values, such that $\Gamma\vdash^c M=\widetilde{M}:\ct{B}$. Moreover, both the CBV and CBN translations only produce complex-value-free computations.\end{theorem}
The reason we had not excluded complex values in the first place is that we would like to have them in our dependently typed generalization of CBPV, where, even though we again exclude them from the computations on which we define the operational semantics, they may occur in their types.
\\
\\
We present a small-step operational semantics for (complex value-free) CBPV computations in terms of a simple abstract machine that Levy calls the CK-machine. The configuration of such a machine consists of a pair $M,K$ where $\Gamma\vdash^c M:\ct{B}$ is a \emph{complex value-free computation} and $\Gamma;\ct{B}\vdash^k K:\ct{C}$ is a compatible (simple) \emph{stack}. We call $\ct{C}$ the type of the configuration. Stacks\footnote{To be precise, these are what Levy calls simple stacks, which are the only stacks necessary to formulate our operational semantics. Analogous to the case of values, one can also conservatively extend the calculus with so-called \emph{complex stacks}, for instance by allowing pattern matching into stacks. As their only purpose is to obtain a precise 1-1 correspondence between categorical models and syntactic theories, they go beyond the scope of this short paper. The reader can find a detailed discussion in \cite{levy2012call}.} are formed according to the rules of figure \ref{fig:simplestacks}.
\begin{figure}[!ht]
\fbox{
\parbox{\linewidth}{
\begin{tabular}{ll}
\AxiomC{}
\UnaryInfC{$\Gamma;\ct{C}\vdash^k \nil : \ct{C}$}
\DisplayProof
& 
\AxiomC{$\Gamma,x:A\vdash^c M:\ct{B}$}
\AxiomC{$\Gamma;\ct{B}\vdash^k K:\ct{C}$}
\BinaryInfC{$\Gamma;FA\vdash^k \toin{[\cdot]}{x}{M}::K:\ct{C}$}
\DisplayProof
\\
&\\
\AxiomC{$\Gamma;\ct{B_j}\vdash^k K:\ct{C}$}
\UnaryInfC{$\Gamma;\Pi_{1\leq i\leq n}\ct{B_i}\vdash^k j::K:\ct{C}$}
\DisplayProof\hspace{80pt}
&
\AxiomC{$\Gamma\vdash^v V:A$}
\AxiomC{$\Gamma;\ct{B}\vdash^k K:\ct{C}$}
\BinaryInfC{$\Gamma;A\Rightarrow \ct{B}\vdash^k V::K:\ct{C}$}
\DisplayProof
\end{tabular}
}
}
\caption{\label{fig:simplestacks} The rules for forming (simple) stacks.}
\end{figure}

The initial configurations, transitions (which embody directed versions of the $\beta$-rules of our equational theory) and terminal configurations in the evaluation of a computation $\Gamma\vdash^c M:\ct{C}$ on the CK-machine are specified by figure \ref{fig:ckmachine}.\\
\quad\\
We recall the following basic results about this operational semantics from \cite{levy2012call,levy2006call}.
\begin{theorem}[Determinism, Strong Normalization and Subject Reduction] For every configuration of the CK-machine, at most one transition applies. No transition applies precisely when the configuration is terminal. Every configuration of type $\ct{C}$ reduces, in a finite number of transitions, to a unique terminal configuration of type $\ct{C}$.
\end{theorem}
\pagebreak
\begin{figure}[t!]
\fbox{
\parbox{\linewidth}{
\textbf{Initial configuration}\\
\begin{tabular}{lll}
$M$ &,& $\nil$
\end{tabular}\\
\\
\textbf{Transitions}\\
\begin{tabular}{lllllll}
$\lbi{V}{x}{M}$ &,& $K$& $\leadsto\hspace{20pt}$ & $M[V/x]$ &,& $K$ \\
$ \toin{M}{x}{N} $ &,& $K$& $\leadsto$ & $ M$ &,& $\toin{[\cdot]}{x}{N}::K$ \\
$   \return V     $ \hspace{20pt}&,& $\toin{[\cdot]}{x}{N}::K$& $\leadsto$ & $   N[V/x]       $ &,& $K$ \\
$   \force\thunk M     $ &,& $K$& $\leadsto$ & $   M       $ &,& $K$ \\
$ \sipm{\langle j,V\rangle}{i}{x}{M_i}      $ &,& $K$& $\leadsto$ & $   M_j[V/x]       $ &,& $K$ \\
$ \upm{\langle\rangle}{M}       $ &,& $K$& $\leadsto$ & $ M         $ &,& $K$ \\
$  \sipm{\langle V,W\rangle }{x}{y}{M}      $ &,& $K$& $\leadsto$ & $  M[V/x,W/y]        $ &,& $K$ \\
$ j\textquoteleft M       $ &,& $K$& $\leadsto$ & $       M   $ &,& $j::K$ \\
$ \lambda_i M_i       $ &,& $j::K$& $\leadsto$ & $    M_j      $ &,& $K$ \\
$ V\textquoteleft M       $ &,& $K$& $\leadsto$ & $      M   $ &,& $V::K$ \\
$ \lambda_x M       $ &,& $V::K$& $\leadsto$ & $   M[V/x]       $ &,& $K$ 
\end{tabular}\\
\\
\textbf{Terminal Configurations}\\
\begin{tabular}{lll}
$\return V$ &,& $\nil$\\
$\lambda_i M_i$ &,& $\nil$\\
$\lambda_x M$ &,& $\nil$\\
$\force z$ &,& $ K$\\
$\sipm{z}{i}{x}{M_i}$ &,& $K$\\
$\upm{z}{M}$ &,& $K$\\
$\sipm{z}{x}{y}{M}$ &,& $K$
\end{tabular}
}
}
\caption{\label{fig:ckmachine} The behaviour of the CK-machine in the evaluation of a computation $\Gamma\vdash^c M:\ct{C}$. We leave out type annotations.\vspace{-15pt}}\vspace{-5pt}
\end{figure}\mbox{}
\clearpage
\subsection{Adding Effects}\label{sec:simpleff}
We recall by example how one adds effects to CBPV. Figure \ref{fig:effects} gives some examples of effects one could add to CBPV, from left to right, top to bottom:  divergence, recursion, printing an element $m$ of some monoid $\mathcal{M}$, erratic choice from finitely many alternatives, errors $e$ from some set $E$, writing a global state $s\in S$ and reading a global state to $s$. We note that the framework fits many more examples like probabilistic erratic choice, local references and control operators \cite{levy2012call}.
\begin{figure}[!ht]
\fbox{
\parbox{\linewidth}{
\begin{tabular}{llll}
\AxiomC{}
\UnaryInfC{$\Gamma\vdash^c \diverge :\ct{B}$}
\DisplayProof
&\AxiomC{$\Gamma,z:U\ct{B}\vdash^c M : \ct{B}$}
\UnaryInfC{$\Gamma\vdash^c \mu_z M : \ct{B}$}
\DisplayProof &

\AxiomC{$\Gamma\vdash^c M:\ct{B}$}
\UnaryInfC{$\Gamma\vdash^c \print{m}M:\ct{B}$}
\DisplayProof
&
\AxiomC{$\{\Gamma\vdash^c M_i:\ct{B}\}_{1\leq i\leq n}$}
\UnaryInfC{$\Gamma\vdash^c \nondet{i}{M_i}:\ct{B}$}
\DisplayProof  \\
&
&
&\\
&&&\\
\AxiomC{}
\UnaryInfC{$\Gamma\vdash^c \error{e} :\ct{B}$}
\DisplayProof
& &
\AxiomC{$\Gamma\vdash^c M:\ct{B}$}
\UnaryInfC{$\Gamma\vdash^c \writecell{s}M:\ct{B}$}
\DisplayProof
&
\AxiomC{$\{\Gamma\vdash^c M_s:\ct{B}\}_{s\in S}$}
\UnaryInfC{$\Gamma\vdash^c \readcell{s}{M_s}:\ct{B}$}
\DisplayProof

\end{tabular}
}
}
\caption{\label{fig:effects} Some examples of effects we could add to CBPV. $\mu_z$ is a name binding operation that binds the identifier $z$ and $\nondet{i}{}$ and $\readcell{s}{}$ bind the indices $i$ and $s$ respectively.}
\end{figure}\quad\\
The small-step semantics of divergence, recursion, erratic choice and errors can easily be explained on our CK-machine as it is. This is summed up in figure \ref{fig:opsemdivs}.\\

\begin{figure}[!ht]
\fbox{
\parbox{\linewidth}{
\textbf{Transitions}\\
\begin{tabular}{lllllll}
$\diverge$ &,& $K$& $\leadsto\hspace{20pt}$ & $\diverge$ &,& $K$ \\
$\mu_z M $ & , & $K$ & $ \leadsto\hspace{20pt}$ & $M[\thunk \mu_z M /z]$ & , & $K$\\
$\nondet{i}{M_i}$ & , & $K$ & $ \leadsto\hspace{20pt}$ & $M_j$ & , & $K$\\
\end{tabular}\\
\\
\textbf{Terminal Configurations}\\
\begin{tabular}{lll}
$\error e$ & , & $K$ 
\end{tabular}
}
}
\caption{\label{fig:opsemdivs} The operational semantics for divergence, recursion, erratic choice and errors.}
\end{figure}\quad\\
For the operational semantics of printing and state, we need to add some hardware to our machine. For that purpose, a configuration of our machine will now consist of a quadruple $M,K,m,s$ where $M,K$ are as before, $m$ is an element of our printing monoid $(\mathcal{M},\epsilon,*)$ which models some channel for output and $s$ is an element of our finite pointed set of states $(S,s_0)$ which is the current value of our storage cell. We lift the operational semantics of all existing language constructs to this setting by specifying that they do not modify $m$ and $s$, that terminal configurations can have any value of $m$ and $s$ and that initial configurations always have value $m=\epsilon$ and $s=s_0$ for the fixed initial state $s_0$. Printing and writing and reading the state can now be given the operational semantics of figure \ref{fig:opsemprint}.
 
\begin{figure}[!ht]
\fbox{
\parbox{\linewidth}{
\textbf{Transitions}\\
\begin{tabular}{lllllllllllllll}
$\print n M$ &,& $K$& ,& $m$& , & $s$ & $\leadsto\hspace{20pt}$ & $M$ &,& $K$ & ,& $m*n$& , & $s$\\
$\writecell {s'} M$ &,& $K$& ,& $m$& , & $s$ & $\leadsto\hspace{20pt}$ & $M$ &,& $K$ & ,& $m$& , & $s'$\\
$\readcell {s'} {M_{s'}}$ &,& $K$& ,& $m$& , & $s$ & $\leadsto\hspace{20pt}$ & $M_s$ &,& $K$ & ,& $m$& , & $s$\\
\end{tabular}
}
}
\caption{\label{fig:opsemprint} The operational semantics for printing and writing and reading global state.}
\end{figure} 
\quad\\
We can try to extend the results of the previous section to this effectful setting and indicate when they break \cite{levy2012call}.
\begin{theorem}[Determinism, Strong Normalization and Subject Reduction] Every transition respects the type of the configuration. No transition occurs precisely if we are in a terminal configuration. In absence of erratic choice, at most one transition applies to each configuration. In absence of divergence and recursion, every configuration reduces to a terminal configuration in a finite number of steps.
\end{theorem}

We can again translate effectful CBV- and CBN-$\lambda$-calculi into CBPV with the appropriate effects as is indicated in figure \ref{fig:transeff}.\\

\begin{figure}[!ht]
\fbox{
\parbox{\linewidth}{
\begin{tabular}{l|ll||ll|l}
\textbf{CBV Term} $M$ & \textbf{CBPV Term} $M^v$&\qquad\qquad\qquad && \textbf{CBN Term} $M$ & \textbf{CBPV Term} $M^n$\\
\hline
$\diverge $& $\diverge$ &&& $\diverge $& $\diverge$\\
$\mu_x M $ & $\mu_z (\toin{\force z}{x}{M^v})$ &&&$\mu_z M $ & $\mu_z M^n$  \\
$\nondet{i}{M_i}$ & $\nondet{i}{M_i^v}$ &&& $\nondet{i}{M_i}$ & $\nondet{i}{M_i^n}$\\
$\error{e} $ & $\error{e}$ &&& $\error{e} $ & $\error{e}$ \\
$\print{m}M$ & $\print{m}M^v$ &&&$\print{m}M$ & $\print{m}M^n$ \\
$\writecell{s}M$ & $\writecell{s}M^v$ &&& $\writecell{s}M$ & $\writecell{s}M^n$\\
$\readcell{s}{M_s}$& $\readcell{s}{M_s^v}$ &&& $\readcell{s}{M_s}$& $\readcell{s}{M_s^n}$
\end{tabular}}}
\caption{\label{fig:transeff} The CBV- and CBN-translations for effectful terms. $z$ is assumed to be fresh in the CBV-translation $\mu_x M$.}
\end{figure}

Let us write $M\Downarrow N,m,s$ for a closed term $\vdash^c M:\ct{B}$ if $M,\nil,\epsilon,s_0$ reduces to the terminal configuration $N,\nil,m,s$. We call this the \emph{big-step semantics} of CBPV. Recall that at least for terms of ground type CBPV induces the usual operational semantics via the CBV- and CBN-translations \cite{levy2006call}.
\begin{theorem}
The big-step semantics for CBPV induces the usual CBV- and CBN-big-step semantics for terms of ground  type, via the respective translations.\end{theorem}

Although one could write down an equational theory for these effects and a corresponding categorical semantics, in which case one would obtain soundness and completeness properties for the CBV- and CBN-translations, we will choose not to do so here for reasons of space. For this, we refer the reader for instance to \cite{plotkin2002notions,levy2012call}. Instead, we just list some basic equations we would typically demand for the effects we consider in figure \ref{fig:effeqn}. The specific equation for recursion will also turn out to be necessary for guaranteeing subject reduction in dCBPV+ which we consider in section \ref{sec:depcbpvklei}. The other equations are mostly relevant as they increase the power of the type checker in the dependently typed case (but complicate a type checking algorithm, of course!). In a practical implementation, one would have to decide which extra equations for the specific effects (like the lookup-update algebra equations for global state of Plotkin and Power \cite{plotkin2002notions}) to include such that one can still design a suitable type checker.

\begin{figure}
[!ht]
\fbox{
\parbox{\linewidth}{
\begin{tabular}{ll}\parbox{.5\linewidth}{
$\toin{\diverge}{x}{M}=\diverge$\\
$\toin{\error{e}}{x}{M}=\error{e}$\\
$\toin{\nondet{i}{N_i}}{x}{M}=\nondet{i}{\toin{N_i}{x}{M}}$\\
$\toin{\readcell{s}{N_s}}{x}{M}=\readcell{s}{\toin{N_s}{x}{M}}$\\
$\toin{(\print{m}{N})}{x}{M}=\print{m}{(\toin{N}{x}{M})}$\\
$\toin{(\writecell{s}{N})}{x}{M}=\writecell{s}{(\toin{N}{x}{M})}$\\
$\mu_z M = M[\thunk\mu_zM/z$} &
\parbox{.5\linewidth}{$\diverge=\lambda_\alpha\diverge$\\
$\error{e}=\lambda_\alpha\error{e}$\\
$\nondet{i}{\lambda_\alpha N_i}=\lambda_\alpha\nondet{i}{N_i}$\\
$\readcell{s}{\lambda_\alpha N_s}=\lambda_\alpha\readcell{s}{N_s}$\\
$\print{m}{\lambda_\alpha N}=\lambda_\alpha\print{m}{N}$\\
$\writecell{s}{\lambda_\alpha N}=\lambda_\alpha \writecell{s}{N}$\\
\;}
\end{tabular}}}
\caption{\label{fig:effeqn} Some basic equations for our example effects, where we use the shorthand notation $\alpha$ for both the case of an identifier $x$ and an index $i$.}
\end{figure}
\clearpage
\section{Dependently Typed Call-by-Push-Value without Dependent Kleisli Extensions}\label{sec:depcbpvwoklext}
	In this section, we sketch how the results in the previous section have an elegant dependently typed generalization, by allowing types to depend on values. We first consider a system in which we only allow sequencing $\toin{M}{x}{N}$ of a dependent function $N$ if its result type does not depend on the identifier $x$ that the result of $M$ is bound to. In other words, we do not include a Kleisli extension principle for dependent functions. The discussion of such a more complicated system where we do include those rules will be postponed to section \ref{sec:depcbpvklei}.
\vspace{-2pt}
\subsection{Syntax}\vspace{-2pt}
The syntax of CBPV generalises straightforwardly to dependent types. As anticipated already by Levy \cite{levy2012call}, we only need to take care in the rule for $\toin{M}{x}{N}$. He suggested that the type $\ct{B}$ should not depend on $x$ in this rule. We shall apply this restriction as well for the moment. We call the resulting system \emph{dependently typed call-by-push-value without dependent Kleisli extensions}, or dCBPV-. We shall later revisit this assumption and study a system in which we do allow such \emph{Kleisli extensions for dependent functions}.

The key feature of a dependent type system is that we allow types to refer contain free identifiers from the context. The reader may want to keep in mind the analogy that dependent types are to predicates what non-dependent types are to propositions. One consequence is that order in the context becomes important as all free variables in a type need to be declared in the context to its left. As types can depend on terms in a dependently typed system, we define both together in one big inductive definition. We refer the reader to \cite{hofmann1997syntax} for a discussion of the syntactic subtleties of a dependent type theory.  

We distinguish between the following objects: contexts $\Gamma;\Delta$, where $\Gamma$ is a region consisting of identifier declarations of value types and $\Delta$ is a region for declarations of computation type and where we write $\Gamma$ as a shorthand for $\Gamma;\cdot$, value types $A$, computation types $\ct{B}$, values $V$, computations $M$ and stacks $K$. The type theory talks about these objects according to the judgements of figure \ref{fig:judgements}.
\vspace{-4pt}
\begin{figure}[!ht]
\fbox{\parbox{\textwidth}{
\begin{tabular}{ll}
\textbf{Judgement} & \textbf{Intended meaning}\vspace{2pt}\\
$\vdash \Gamma;\Delta \;\ctxt$ & $\Gamma;\Delta$ is a valid context\\
$\Gamma \vdash A\;\vtype$ &  $A$ is a value type in context $\Gamma$\\
$\Gamma \vdash \ct{B}\;\ctype$ &  $\ct{B}$ is a computation type in context $\Gamma$\\
$\Gamma\vdash^v V:A$ & $V$ is a value of type $A$ in context $\Gamma$\\
$\Gamma\vdash^c M:\ct{B}$ & $M$ is a computation of type $\ct{B}$ in context $\Gamma$\\
$\Gamma;\Delta\vdash^k K:\ct{B}$ & $K$ is a stack of type $\ct{B}$ in context $\Gamma;\Delta$\\
$\vdash \Gamma;\Delta = \Gamma';\Delta'$\hspace{40pt} & $\Gamma;\Delta$ and $\Gamma';\Delta'$ are judgementally equal contexts\\
$\Gamma\vdash A=A'$ & $A$ and $A'$ are judgementally equal value types in context $\Gamma$\\
$\Gamma\vdash \ct{B}= \ct{B'}$ & $\ct{B}$ and $\ct{B'}$ are judgementally equal computation types in context $\Gamma$\\
$\Gamma\vdash^v V= V':A$ & $V$ and $V'$ are judgementally equal values of type $A$ in context $\Gamma$\\
$\Gamma\vdash^c M= M':\ct{B}$ & $M$ and $M'$ are judgementally equal computations of type $\ct{B}$ in context $\Gamma$\\
$\Gamma;\Delta\vdash^k K= K':\ct{B}$ & $K$ and $K'$ are judgementally equal stacks of type $\ct{B}$ in context $\Gamma;\Delta$
\end{tabular}}}
\normalsize
\caption{\label{fig:judgements} Judgements of dependently typed CBPV.\vspace{-10pt}}
\end{figure}

To start with, we have rules, which we shall not list, which state that all judgemental equalities are equivalence relations and that all term, type and context constructors as well as substitutions respect judgemental equality. In similar vein, we have conversion rules which state that we may swap contexts and types for judgementally equal ones in all judgements. To form contexts, we have the rules of figure~\ref{fig:ctxtrules}.\vspace{-4pt}\nopagebreak
\begin{figure}[!ht]
\fbox{
\parbox{\linewidth}{
\begin{tabular}{ll}
\AxiomC{}
\UnaryInfC{$\cdot;\cdot \ctxt$}
\DisplayProof
& \\
& \\
\AxiomC{$\vdash\Gamma;\Delta\ctxt$}
\AxiomC{$\Gamma\vdash A\;\vtype$}
\BinaryInfC{$\vdash \Gamma,x:A;\Delta \ctxt$}
\DisplayProof

&
\hspace{56pt}
\AxiomC{$\vdash \Gamma;\cdot\ctxt$}
\AxiomC{$\Gamma \vdash \ct{B}\ctype$}
\BinaryInfC{$\vdash \Gamma;\ct{B}\ctxt$}
\DisplayProof\\
\end{tabular}
}
}
\caption{\label{fig:ctxtrules} Rules for contexts, where $x$ is assumed to be a fresh identifier.\vspace{-40pt}\;}
\end{figure}
To form types, we have the rules of figure \ref{fig:depcbpvtypes}.
\begin{figure}[!ht]
\fbox{
\parbox{\linewidth}{
\begin{tabular}{ll}
\AxiomC{$\Gamma,x:A,\Gamma'\vdash A'\vtype$}
\AxiomC{$\Gamma\vdash^v V:A$}
\BinaryInfC{$\Gamma,\Gamma'[V/x]\vdash A'[V/x]\vtype$}
\DisplayProof\hspace{40pt}
&
\AxiomC{$\Gamma,x:A,\Gamma'\vdash \ct{B}\ctype$}
\AxiomC{$\Gamma\vdash^v V:A$}
\BinaryInfC{$\Gamma,\Gamma'[V/x]\vdash \ct{B}[V/x]\ctype$}
\DisplayProof\\
&\\
\AxiomC{$\Gamma\vdash \ct{B}\ctype$}
\UnaryInfC{$\Gamma\vdash U\ct{B}\vtype$}
\DisplayProof
& 
\AxiomC{$\Gamma\vdash A\vtype$}
\UnaryInfC{$\Gamma\vdash FA\ctype$}
\DisplayProof
\\
& \\
\AxiomC{$\{\Gamma\vdash A_i\vtype\}_{1\leq i \leq n}$}
\UnaryInfC{$\Gamma\vdash \Sigma_{1\leq i\leq n}A_i\vtype$}
\DisplayProof

&
\AxiomC{$\{\Gamma\vdash \ct{B_i}\ctype\}_{1\leq i \leq n}$}
\UnaryInfC{$\Gamma\vdash \Pi_{1\leq i\leq n}\ct{B_i}\ctype$}
\DisplayProof\\
&\\
\AxiomC{$\Gamma,x:A\vdash A'\vtype$}
\UnaryInfC{$\Gamma\vdash \Sigma_{x:A}A'\vtype$}
\DisplayProof

&
\AxiomC{$\Gamma,x:A\vdash \ct{B}\ctype$}
\UnaryInfC{$\Gamma\vdash \Pi_{x:A}\ct{B}\ctype$}
\DisplayProof\\
&\\
\AxiomC{$\vdash \Gamma\ctxt$}
\UnaryInfC{$\Gamma\vdash 1\vtype$}
\DisplayProof
&\\
&\\
\AxiomC{$\Gamma\vdash^v V:A$}
\AxiomC{$\Gamma\vdash^v V':A$}
\BinaryInfC{$\Gamma\vdash \Id_A(V,V')\vtype$}
\DisplayProof

\end{tabular}
}
}
\caption{\label{fig:depcbpvtypes} Rules for type formation.}
\end{figure}\\
For these types, we consider the values and computations formed using the rules of figure \ref{fig:vcdepterms}. We leave the discussion of stacks until section \ref{sec:depop}.\\
\begin{figure}[!ht]
\centering
\fbox{\parbox{\linewidth}{
\begin{tabular}{ll}
\AxiomC{$\vdash \Gamma,x:A,\Gamma'\ctxt$}
\UnaryInfC{$\Gamma,x:A,\Gamma'\vdash^v x:A$}
\DisplayProof\hspace{50pt} & 
\AxiomC{$\Gamma\vdash^v V:A$}
\AxiomC{$\Gamma,x:A,\Gamma'\vdash^{v/c} R:{B}$}
\BinaryInfC{$\Gamma,\Gamma'[V/x]\vdash^{v/c} \lbi{x}{V}{R} :{B}[V/x]$}
\DisplayProof\\
&\\
\AxiomC{$\Gamma\vdash^v V:A$}
\UnaryInfC{$\Gamma\vdash^c \return\; V:FA$}
\DisplayProof &
\AxiomC{$\Gamma\vdash^c M:FA$}
\AxiomC{$\Gamma,x:A,\Gamma'\vdash^c N:\ct{B}$}
\AxiomC{$x$ not free in $\Gamma';\ct{B}$}
\TrinaryInfC{$\Gamma,\Gamma' \vdash^c \toin{M}{x}{N}:\ct{B}$}
\DisplayProof\\
&\\
\AxiomC{$\Gamma\vdash^c M:\ct{B}$}
\UnaryInfC{$\Gamma\vdash^v \thunk M:U\ct{B}$}
\DisplayProof
&
\AxiomC{$\Gamma\vdash^v V: U\ct{B}$}
\UnaryInfC{$\Gamma\vdash^c \force V: \ct{B}$}
\DisplayProof\\
&\\
\AxiomC{$\Gamma\vdash^v V_i: A_i$}
\UnaryInfC{$\Gamma\vdash^v \langle i,V_i\rangle : \Sigma_{1\leq i\leq n}A_i$}
\DisplayProof
&
\AxiomC{$\Gamma\vdash^v V: \Sigma_{1\leq i\leq n}A_i$}
\AxiomC{$\{\Gamma,x:A_i\vdash^{v/c} R_i : {B}[\langle i,x\rangle/z]\}_{1\leq i\leq n}$}
\BinaryInfC{$\Gamma\vdash^{v/c} \sipm{V}{i}{x}{R_i} : {B}[V/z]$}
\DisplayProof\\
&\\
\AxiomC{$\vdash \Gamma\ctxt$}
\UnaryInfC{$\Gamma\vdash^v\langle\rangle :1$}
\DisplayProof&
\AxiomC{$\Gamma\vdash^v V:1$}
\AxiomC{$\Gamma\vdash^{v/c} R:{B}[\langle\rangle/z]$}
\BinaryInfC{$\Gamma\vdash^{v/c} \upm{V}{R}:{B}[V/z]$}
\DisplayProof
\\
&\\
\AxiomC{$\Gamma\vdash^v V_1:A_1$}
\AxiomC{$\Gamma\vdash^v V_2:A_2[V_1/x]$}
\BinaryInfC{$\Gamma\vdash^v \langle V_1,V_2\rangle :\Sigma_{x:A_1} A_2$}
\DisplayProof
&
\AxiomC{$\Gamma\vdash^v V: \Sigma_{x:A_1} A_2$}
\AxiomC{$\Gamma,x:A_1,y:A_2\vdash^{v/c} R:{B}[\langle x,y\rangle/z]$}
\BinaryInfC{$\Gamma\vdash^{v/c} \sipm{V}{x}{y}{R}:{B}[V/z]$}
\DisplayProof\\
&\\
\AxiomC{$\Gamma\vdash^v V:A$}
\UnaryInfC{$\Gamma\vdash^v \refl{V}:\Id_A(V,V)$}
\DisplayProof &
\AxiomC{$\Gamma\vdash^v V:\Id_A(V_1,V_2)$}
\AxiomC{$ \Gamma,x:A\vdash^{v/c} R :B[x/x',\refl{x}/p]$}
\BinaryInfC{$\Gamma\vdash^{v/c} \idpm{V}{x}{R}:B[V_1/x,V_2/x',V/p]$}
\DisplayProof \\
&\\
\AxiomC{$\{\Gamma\vdash^c M_i :\ct{B_i}\}_{1\leq i\leq n}$}
\UnaryInfC{$\Gamma\vdash^c \lambda_i M_i : \Pi_{1\leq i\leq n}\ct{B_i}$}
\DisplayProof
&
\AxiomC{$\Gamma\vdash^c M: \Pi_{1\leq i\leq n}\ct{B_i}$}
\UnaryInfC{$\Gamma\vdash^c i\textquoteleft M : \ct{B_i}$}
\DisplayProof\\
&\\
\AxiomC{$\Gamma,x:A\vdash^c M:\ct{B}$}
\UnaryInfC{$\Gamma\vdash^c \lambda_xM:\Pi_{x:A}\ct{B}$}
\DisplayProof
&
\AxiomC{$\Gamma\vdash^v V:A$}
\AxiomC{$\Gamma\vdash^c M:\Pi_{x:A}\ct{B}$}
\BinaryInfC{$\Gamma\vdash^c V\textquoteleft M : \ct{B}[V/x]$}
\DisplayProof
\end{tabular}
}
}
\caption{\label{fig:vcdepterms} Values and computations of dependently typed CBPV.}
\end{figure}

We generate judgemental equalities for values and computations through the rules of figure \ref{fig:vceqs} and \ref{fig:vcdepeqs}. Note that we are using extensional $\Id$-types, in the sense of $\Id$-types with an $\eta$-rule. This is only done for the aesthetics of the categorical semantics. They may not be suitable for an implementation, however, as they make type checking undecidable for the usual reasons \cite{hofmann1997syntax}. The syntax and semantics can just as easily be adapted to intensional $\Id$-types, which are the obvious choice to for an implementation.

\begin{figure}[!ht]
\fbox{
\parbox{\linewidth}
{\resizebox{1.02\linewidth}{!}{
\begin{tabular}{ll}
\hspace{-6pt}$\idpm{(\refl V)}{x}{R} = R[V/x]$ \hspace{0pt}& $R[V_1/x,V_2/y,V/z] \stackrel{\#w}{=} \idpm{V}{w}{R[w/x,w/y,(\refl w)/z]}$
\end{tabular}
}}
}
\caption{\label{fig:vcdepeqs} Equations for computations involving reflexivity witnesses. Again, these rules should be read as equations of typed terms in context: they are assumed to hold if we can derive that both sides of the equation are terms of the same type in the same context.}
\end{figure}

\vspace{-5pt}
Figures \ref{fig:depcbvtrans} and \ref{fig:depcbntrans} indicate the natural candidate CBV- and CBN-translations of DTT into dCBPV.

\begin{figure}[!ht]
\fbox{
\parbox{\linewidth}{
\begin{tabular}{l|l||l|l}
\textbf{CBV type}  & \textbf{CBPV type} & \textbf{CBV term } & \textbf{CBPV term}\\
\hline
$\Gamma\vdash A\type$ & $\vect{UF}\Gamma^v\vdash A^v\vtype$ & $x_1:A_1,\ldots,x_m:A_m$ & $x_1:A_1^v,\ldots,x_m:A_m^v[\ldots \tr x_{i}/z_{i}\ldots] $\\
 & & $\vdash M:A$ & $\vdash^c M^v:F(A^v[\tr x_1/z_1,\ldots,\tr x_n/z_n])$\\
 && $x$& $\return x$\\
  &&$\lbi{x}{M}{N}$ & $\toin{M^v}{x}{N^v}$ \\
$\Sigma_{1\leq i\leq n }A_i$ & $\Sigma_{1\leq i\leq n }A_i^v$& $\langle i,M\rangle $&$\toin{M^v}{x}{\return \langle i, x\rangle }$ \\
 &&$\sipm{M}{i}{x}{N_i}$& $\toin{M^v}{z}{(\sipm{z}{i}{x}{N_i^v})}$\\
$\Pi_{1\leq i\leq n}A_i $ & $U\Pi_{1\leq i \leq n} FA_i^v$ &$\lambda_iM_i$ &$\return \thunk (\lambda_i  M_i^v)$\\
&&$i\textquoteleft  N $&$\toin{N^v}{z}{(i\textquoteleft \force z)}$\\
$\Pi_{x:A} A'$ & $U(\Pi_{x:A^v} F A'^v[\tr x/z])$ & $\lambda_x M$&$\return \thunk \lambda_x M^v$\\
&&$M\textquoteleft N$ &$\toin{M^v}{x}{(\toin{N^v}{z}{(x\textquoteleft \force z)})}$\\
$1$ & $1$ & $\langle\rangle$ & $ \return \langle\rangle$  \\
&&$\upm{M}{N}$&$\toin{M^v}{z}{(\upm{z}{N^v})}$\\
$\Sigma_{x:A}  A'$ & $\Sigma_{x:A^v} A'^v[\tr x/z]$ & $ \langle M, N\rangle $  & $\toin{M^v}{x}{(\toin{N^v}{y}{\return \langle x,y\rangle})}$\\
&&$\sipm{M}{x}{y}{N}$&$\toin{M^v}{z}{(\sipm{z}{x}{y}{N^v})}$\\
$\Id_A(M,N)$& {$\Id_{ UF A^v}( \thunk M^v$,} &$\refl M$& $\toin{M^v}{z}{\return \refl \tr z} $ \\
&{$ \thunk N^v)$}&$\idpm{M}{x}{N}$& $\toin{M^v}{z}{(\idpm{z}{y}{}}$\\
&&& $(\toin{\force y}{x}{N^v}))$
\end{tabular}}
}
\caption{\label{fig:depcbvtrans} A translation of dependently typed CBV into dCBPV. We write $\tr$ as an abbreviation for $\thunk\return$ and $\vect{UF}\Gamma:=z_1:UF A_1,\ldots,z_n:UFA_n$ for a context $\Gamma=x_1:A_1,\ldots,x_n:A_n$.}
\end{figure}
\vspace{-15pt}
\begin{figure}[!ht]
\fbox{
\parbox{\linewidth}{
\begin{tabular}{l|l||l|l}
\textbf{CBN type}  & \textbf{CBPV type} & \textbf{CBN term } & \textbf{CBPV term }\\
\hline
${\Gamma}\vdash\ct{B}\type$ &  $\vect{U}{\Gamma^n}\vdash\ct{B}^n\ctype$& $x_1:\ct{B}_1,\ldots,x_m:\ct{B}_m\vdash M:\ct{B}$& $x_1:U\ct{B}_1^n,\ldots,x_m:U\ct{B}_m^n\vdash^c M^n:\ct{B}^n$ \\
 && $x$& $\force x$\\
  & & $\lbi{x}{M}{N}$ & $\lbi{x}{(\thunk M^n)}{N^n}$ \\
$\Sigma_{1\leq i\leq n }\ct{B}_i$ & $F\Sigma_{1\leq i\leq n }U\ct{B}_i^n$& $\langle i,M\rangle $&$\return \langle i,\thunk M^n\rangle $ \\
 & &$\sipm{M}{i}{x}{N_i}$&$\toin{M^n}{z}{(\sipm{z}{i}{x}{N_i^n})}$ \\
$\Pi_{1\leq i\leq n}\ct{B}_i $ & $\Pi_{1\leq i \leq n} \ct{B}_i^n$ & $\lambda_iM_i$& $\lambda_iM_i^n$\\
 && $i\textquoteleft M$ & $i\textquoteleft M^n$\\
$\Pi_{x:\ct{B}} \ct{B'}$ & $\Pi_{x:U\ct{B}^n} \ct{B'}^n$ & $\lambda_x M $& $\lambda_xM^n$\\
 &&$N\textquoteleft M$ & $(\thunk N^n) \textquoteleft M^n$ \\
$1$ & $F1$ & $\langle\rangle$ & $\return \langle\rangle$  \\
&&$\upm{M}{N}$&$\toin{M^n}{z}{(\upm{z}{N^n})}$\\
$\Sigma_{x:\ct{B}} \ct{B'}$ & $F(\Sigma_{x:U\ct{B}^n}  U\ct{B'}^n)$ & $\langle M, N\rangle $  & $\return \langle \thunk M^n,\thunk N^n\rangle$\\
&& $\sipm{M}{x}{y}{N}$& $\toin{M^n}{z}{(\sipm{z}{x}{y}{N^n})}$ \\
$\Id_{\ct{B}}(M,M')$&$F(\Id_{U\ct{B}}(\thunk M^n$ & $\refl M$& $\return \refl \thunk M^n$\\
&$,\thunk M'{}^n))$&$\idpm{M}{x}{N}$& $\toin{M^n}{z}{(\idpm{z}{x}{N^n})}$
\end{tabular}
}
}
\caption{\label{fig:depcbntrans} A translation of dependently typed CBN into dCBPV. We write $\vect{U}\Gamma:=x_1:U A_1,\ldots,x_n:UA_n$ for a context $\Gamma=x_1:A_1,\ldots,x_n:A_n$.}
\end{figure}\quad\\
However, it turns out that without dependent Kleisli extensions, the CBV-translation is not well-defined as it results in untypable terms. The CBN-translation is, but only if we restrict to the weak (non-dependent) elimination rules for $\Sigma_{1\leq i\leq n}$, $1$-, $\Sigma$- and $\Id$-types, meaning that the type we are eliminating into does not depend on the type being eliminated from. For an alternative to the CBV-translation, we would expect the CBV-translation to factorise as a translation into a dependently typed equivalent of Moggi's' monadic metalanguage, followed by a translation from this monadic language into dCBPV. It is, in fact, the former that is ill-defined if we do not have a principle of Kleisli extensions in our monadic language (or, correspondingly, in dCBPV). What we can define  is a translation from a dependently typed monadic language (without dependent Kleisli extensions) into dCBPV-. In this case, we can, in fact, use the strong (dependent) elimination rules for all positive connectives. Perhaps this is a (partial) explanation of why all (CBV) dependently typed languages with effects have encapsulated the effects in a monad. The exceptions are non-termination and recursion. As we shall see in section \ref{sec:depcbpvklei}, dependent Kleisli extensions are entirely unproblematic in that case, which means we can treat these examples as first class effects in a dependently typed language and we do not have to encapsulate them in a modality.

By analogy with the simply typed scenario, it seems very likely that one would be able to state soundness and completeness results for these translations, if one used the canonical equational theories for CBV- and CBN-dependent type theory. As we are not aware of any such equational theories being described in literature, we propose to \emph{define} the CBV- and CBN-equational theory on dependent type theories through their translations into CBPV.

\clearpage
\subsection{Categorical Semantics}
We have now reached the point in the story that was our initial motivation to study dependently typed CBPV: its very natural categorical semantics. Note that we have the following elegant generalization of our reformulated notion of categorical model for simple CBPV.
\begin{definition}[dCBPV- Model] By a categorical model of dCBPV-, we shall mean the following data.
\begin{itemize}
\item an indexed category $\Bcat^{op}\ra{\Ccat}\Cat$ of \emph{values} with full and faithful democratic comprehension;
\item an indexed category $\Bcat^{op}\ra{\Dcat}\Cat$ of \emph{stacks};\\
\item $0,+$-types in $\Ccat$ such that, additionally, the following induced maps are bijections:
$$\Dcat(C.\Sigma_{1\leq i\leq n} C_i )(\ct{D},\ct{D'})\ra{}\Pi_{1\leq i \leq n}\Dcat(C.C_i)(\ct{D}\{\proj{C}{\langle i,\id_{C_i}\rangle }\},\ct{D'}\{\proj{C}{\langle i,\id_{C_i}\rangle }\});$$
\item an indexed adjunction \mbox{
\begin{diagram}
\Dcat & \pile{\lTo^F\\\bot\\\rTo_U} &\Ccat;
\end{diagram}}
\item $\Pi$-types in $\Dcat$ in the sense of having right adjoint functors $-\{\proj{A}{B}\}\dashv \Pi_B:\Dcat(A)\ra{}\Dcat(A.B)$ satisfying the right Beck-Chevalley condition;
\item Finite indexed products $(\top,\&)$ in $\Dcat$;
\item $\Sigma$-types in $\Ccat$;
\item $\Id$-types in $\Ccat$\footnote{In case we work with intensional $\Id$-types, we probably want to add the additional condition, which corresponds to complex stacks, that says that the canonical map $\Dcat(A.A'.A'.\Id_{A'})(\ct{B},\ct{B'})\ra{}\Dcat(A.A')(\ct{B}\{\langle \diag{A}{A'},\mathsf{refl}_{A'}\rangle\},\ct{B'}\{\langle \diag{A}{A'},\mathsf{refl}_{A'}\rangle\})$ is a retraction. This map is automatically an isomorphism in our case of extensional $\Id$-types.}.
\end{itemize}
\end{definition}
Again, this semantics is sound and complete and we could in fact strengthen the completeness theorem below, if we include complex stacks in the syntax, to a 1-1 theory-model correspondence as is done in \cite{levy2005adjunction}.
\begin{theorem}[dCBPV- Semantics] We have a sound interpretation of dCBPV- in a dCBPV- model:
\begin{align*}
\sem{\cdot} & = \cdot\\
\sem{\Gamma,x:A} &= \sem{\Gamma}.\sem{A}\\
\sem{\Gamma\vdash^v A}&=\Ccat(\sem{\Gamma})(1,\sem{A})\\
\sem{\Gamma\vdash^c \ct{B}} & = \Dcat(\sem{\Gamma})(F1,\sem{\ct{B}})\\
\sem{\Gamma;\ct{B}\vdash^k \ct{C}} & = \Dcat(\sem{\Gamma})(\sem{\ct{B}},\sem{\ct{C}})\\
\sem{A[V/x]} & = \sem{A}\{\langle \langle \id_{\sem{\Gamma}},\sem{V}\rangle,\id_{\sem{\Gamma'[V/x]}}\rangle\}\\
\sem{\ct{B}[V/x]} &= \sem{\ct{B}}\{\langle \langle \id_{\sem{\Gamma}},\sem{V}\rangle,\id_{\sem{\Gamma'[V/x]}}\rangle \}\\
\sem{1}&=1\\
\sem{\Sigma_{x:A} A'} &= \Sigma_{\sem{A}} \sem{A'}\\
\sem{\Sigma_{1\leq i\leq n}A_i}&=(\cdots(\sem{A_1}+\sem{A_2})+\cdots)+\sem{A_n})\\
\sem{\Id_A(V,V')} & = \Id_{\sem{A}}\{\langle\langle  \id_{\sem{\Gamma}} , \sem{V}\rangle ,\sem{V'}\rangle \}\\
\sem{\Pi_{1\leq i\leq n}\ct{B}_i}&=(\cdots(\sem{\ct{B}_1}\&\sem{\ct{B}_2})\&\cdots)\&\sem{\ct{B}_n})\\
\sem{\Pi_{x:A} \ct{B}} & = \Pi_{\sem{A}}\sem{\ct{B}}\\
\sem{FA}&=F\sem{A}\\
\sem{U\ct{B}} &= U\sem{\ct{B}},
\end{align*}
where we also give the interpretation of the stack judgement (see section \ref{sec:depop}), together with the obvious interpretation of terms. The interpretation in such categories is complete in the sense that an equality of values or computations holds in all interpretations iff it is provable in the syntax of dCBPV-. In fact, if we add the obvious admissible weakening and exchange rules to dCBPV-, the interpretation defines an onto correspondence between categorical models and syntactic theories in dCBPV- which satisfy mutual soundness and completeness results. This correspondence becomes 1-1 and we obtain completeness for the stack judgement if we include complex stacks.
\end{theorem}
\begin{proof}[Proof (sketch)] The proof goes almost entirely along the lines of the soundness and completeness proofs for linear dependent type theory in \cite{vakar2014syntax}. Nothing surprising happens in the soundness proof. For the completeness result, we build a syntactic category. The thing to note here is that we first conservatively extend our syntax with complex stacks before doing so, as is done in \cite{levy2005adjunction}. \end{proof}
This leads to an induced notion of model for CBN-dependent type theory.
\begin{theorem}[Dependent CBN-Semantics {1}] The (semantic equivalent of the) CBN-translation of DTT with $\Sigma_{1\leq i\leq n}$-, $1$-, $\Sigma$-, $\Id$-, $\Pi_{1\leq i\leq n}$-, $\Pi$-types, where we use the weak (non-dependent) elimination rules for all positive connectives, into dCBPV-, lets us construct a categorical model of CBN-dependent type theory with the connectives above out of any model of dCBPV- by taking the coKleisli (indexed) category for $!:=FU$. The interpretation of CBN-dependent type theory is sound and complete for the equational theory induced from dCBPV-:
\begin{align*}
\sem{\ct{B_1},\cdots,\ct{B_n}\vdash \ct{B}}&=\Dcat(U\sem{\ct{B_1}}.\cdots.U\sem{\ct{B_n}})(F1,\sem{\ct{B}})\cong\Dcat_{!}(U\sem{\ct{B_1}}.\cdots.U\sem{\ct{B_n}})(\top,\sem{\ct{B}}).
\end{align*}
\end{theorem}

We note that this coKleisli category, our notion of a model of CBN-dependent type theory, is very close to the usual notion of a model of pure DTT.
\begin{theorem}[Dependent CBN-Categories] The coKleisli category $\Dcat_{!}$ is an indexed category with full and faithful (possibly undemocratic) comprehension with fibred finite products $\Pi_{1\leq i\leq n}$ as well as $\Pi$-types. It supports weak $\Sigma_{1\leq i\leq n}$-, $\Sigma$- and $\Id$-types (non-dependent elimination rules, no $\eta$-rules).
\end{theorem}
\begin{proof}[Proof (sketch)]
$\Dcat_!$ satisfies the comprehension axiom in the sense that we have homset isomorphism
\begin{align*}
\Dcat_!(\Gamma')(\top,B\{f\})&=\Dcat(\Gamma')(FU\top,B\{f\})\\
&\cong\Dcat(\Gamma')(F1,B\{f\})\\
&\cong \Ccat(\Gamma')(1,U(B\{f\}))\\
&= \Ccat(\Gamma')(1,U(B)\{f\})\\
&\cong \Bcat/\Gamma(f,\proj{\Gamma}{UB}).
\end{align*}
As this is a special case of the comprehension for $\Ccat$, we know it to be full and faithful.

We know from the simply typed case that fibre-wise products in $\Dcat$ give rise to products in $\Dcat_!$. These are stable under change of base by assumption.

Note that $\Pi$-types directly follow as a special case of $\Pi$-types in $\Dcat$:
\begin{align*}
\Dcat_!(\Gamma.UA)(B\{\proj{\Gamma}{UA}\},C)&=\Dcat(\Gamma.UA)(FU(B\{\proj{\Gamma}{UA}\}),C)\\
&=\Dcat(\Gamma.UA)((FUB)\{\proj{\Gamma}{UA}\},C)\\
&\cong\Dcat(\Gamma)(FUB,\Pi_{UA}C)\\
&=\Dcat_!(\Gamma)(B,\Pi_{UA}C).
\end{align*}
For $\Sigma$-types, we note that we have maps back and forth, given by the unit and counit of the adjunction between $F$ and $U$ which satisfy a $\beta$-law given by one of the triangle identities for the adjunction:
\begin{align*}
\Dcat_!(\Gamma.UA)(B,C\{\proj{\Gamma}{UA}\})&=\Dcat(\Gamma.UA)(FU(B),C\{\proj{\Gamma}{UA}\})\\
&\cong\Ccat(\Gamma.UA)(UB,U(C\{\proj{\Gamma}{UA}\}))\\
&=\Ccat(\Gamma.UA)(UB,(UC)\{\proj{\Gamma}{UA}\})\\
&\cong\Ccat(\Gamma)(\Sigma_{UA}UB,UC)\\
&\cong\Dcat(\Gamma)(F\Sigma_{UA}UB,C)\\
&\leftrightarrows\Dcat(\Gamma)(FUF\Sigma_{UA}UB,C)\\
&=\Dcat_!(\Gamma)(F\Sigma_{UA}UB,C).
\end{align*}
The same argument gives us the corresponding statement for $\Sigma_{1\leq i\leq n}$- and $\Id$-types.
\end{proof}
We note that although we started with extensional $\Id$-types in dCBPV-, we have obtained possibly intensional $\Id$-types in dependent CBN.

We postpone the categorical discussion of models for dependently typed CBV until we add dependent Kleisli extensions to dCBPV- in section \ref{sec:depcbpvklei}. For now, we would just like to point out that $\Ccat$ equipped with the indexed monad $T:=UF$ defines what should be regarded as a model of a dependently typed equivalent of Moggi's monadic metalanguage, without dependent context extensions.

\begin{theorem}[Dependent monadic metalanguage models]
Given a model $\Ccat\leftrightarrows \Dcat$ of dCBPV-, $T:=UF$ defines an indexed monad on $\Ccat$, which has a generalized notion of strength $\Sigma_ATB\ra{s_{A,B}}T\Sigma_AB$.
\end{theorem}
\begin{proof}[Proof (sketch)] We note that, starting from $\id_{\Sigma_AB}$, we can obtain a generalised notion of strength for $T$:
\begin{align*}
\Ccat(\Gamma)(\Sigma_A B ,\Sigma_A B) &\cong \Ccat(\Gamma)(B,(\Sigma_AB)\{\mathbf{p}_{\Gamma,A}\})\\
&\ra{T} \Ccat(\Gamma)(TB,T(\Sigma_A B)\{\mathbf{p}_{\Gamma,A}\})\\
&= \Ccat(\Gamma)(TB,(T\Sigma_A B)\{\mathbf{p}_{\Gamma,A}\})\\
&\cong\Ccat(\Gamma)(\Sigma_ATB,T\Sigma_AB).
\end{align*}
In particular (for the case where $\Gamma=\cdot$, using full and faithful comprehension), we get $\Gamma.TA\ra{s_{A,B}}T(\Gamma.A)\in \Bcat$.\end{proof}
\begin{remark}Note that we cannot in general define a costrength $\Sigma_{TA} B\ra{}T\Sigma_A B\{\proj{\Gamma}{\eta_A}\}$ or, therefore, a pairing $\Sigma_{TA}TB\ra{} T\Sigma_A B\{\proj{\Gamma}{\eta_A}\}$. This asymmetry does not occur in the simply typed setting. It will be mended by the addition of Kleisli extensions for dependent functions.\end{remark}

In the simply typed setting, one can factor the CBV-translation from the $\lambda$-calculus into CBPV through the monadic metalanguage. What we see happening in the setting with dependent types is that while the translation from the dependently typed monadic metalanguage with dependent Kleisli extensions in dCBPV- works fine, we cannot define the obvious CBV-translation from dependent type theory into the dependently typed monadic metalanguage, unless we have dependent Kleisli extensions.

\clearpage
\subsection{Some Basic Models}
We can first note that any model of pure dependent type theory is, by using the identity adjunction, in particular, a model of dependently typed CBPV (with or without dependent Kleisli extensions), which shows consistency of the calculus.
\begin{theorem}
Dependently typed CBPV with (dCBPV+) and without (dCBPV-) dependent Kleisli extensions is consistent.
\end{theorem}
More interestingly, any model of intuitionistic linear dependent type theory supporting the appropriate connectives \cite{vakar2014syntax,vakar2015syntax} gives rise to a model of dependently typed CBPV without dependent Kleisli extensions, modelling commutative effects.
\begin{theorem}The notion of a model given by \cite{vakar2014syntax} for the dependently typed linear-non-linear logic of \cite{krishnaswami2015integrating} with the additional connectives of finite additive disjunctions is precisely a dCBPV- model such that we have symmetric monoidal closed structures $(I,\otimes,\multimap)$ on the fibres of $\Dcat$, stable under change of base, ($\Dcat$ is an indexed symmetric monoidal closed category) s.t. $F$ consists of strong monoidal functors (sending nullary and binary products in $\Ccat$ to $I$ and $\otimes$ in $\Dcat$) and which supports $\Sigma_{F-}^\otimes$-types (see section \ref{sec:deec}).
\end{theorem}
As in the simply typed setting, models of pure DTT on which we have an indexed monad are again a source of examples of dCBPV- models.
\begin{theorem}
Let $\Bcat^{op}\ra{\Ccat}\Cat$ be a model of pure DTT (with all type formers discussed) on which we have an indexed monad $T$. Then, the fibre-wise Eilenberg-Moore adjunction $\Ccat\leftrightarrows \Ccat^T$ gives a model of dCBPV-.
\end{theorem}
\begin{proof}[Proof (sketch)] As in the simply typed setting, a product of algebras is just the product of their carriers equipped with the obvious algebra structure. Indeed, it is a basic fact from category theory that the forgetful functor from the Eilenberg-Moore category creates limits. Given an object $TB\ra{k}B$ of $\Ccat^T(\Gamma.A)$, we note that we also obtain a canonical $T$-algebra structure on $\Pi$-types of carriers (starting from the identity on $\Pi_A B$):
\begin{align*}
\Ccat(\Gamma)(\Pi_AB,\Pi_A B)&\cong \Ccat(\Gamma.A)((\Pi_AB)\{\proj{\Gamma}{A}\},B)\\
&\ra{T}  \Ccat(\Gamma.A)(T((\Pi_AB)\{\proj{\Gamma}{A}\}),TB)\\
&\cong  \Ccat(\Gamma.A)((T\Pi_AB)\{\proj{\Gamma}{A}\},TB)\\
&\ra{-;k}  \Ccat(\Gamma.A)((T\Pi_AB)\{\proj{\Gamma}{A}\},B)\\
&\cong \Ccat(\Gamma)(T\Pi_AB,\Pi_AB).
\end{align*}
We leave the verification of the $T$-algebra axioms to the reader. We define the result to be $\Pi_{A}k$. Note that it is precisely defined so that, for an algebra $TC\ra{l}C$, the isomorphism $\Ccat(\Gamma.A)(C\{\proj{\Gamma}{A}\},B)\cong\Ccat(\Gamma)(C,\Pi_A B)$ restricts to $\Ccat^T(\Gamma.A)(l\{\proj{\Gamma}{A}\},k)\cong\Ccat^T(\Gamma)(l,\Pi_k B)$.
\end{proof}

A concrete example to which we can apply the previous theorem is obtained for any monad $T$ on $\Set$. Indeed, we can lift $T$ (point-wise) to an indexed monad on the usual families of sets model $\Fam(\Set)$ of pure DTT\footnote{Recall that $\Fam(\Set)$ is defined as the restriction to $\Set\subseteq \Cat$ of the ($\Cat$-enriched) hom-functor into $\Set$:\linebreak
$\Set^{op}\subseteq \Cat^{op}\ra{\Cat(-,\Set)}\Cat$.}. In a different vein, given a model $\Ccat$ of pure DTT, the usual exception, global state, reader, writer and continuation monads (which we form using objects of $\Ccat(\cdot)$) give rise to indexed monads, hence we obtain models of dCBPV-. More exotic examples are the many indexed monads that arise from homotopy type theory, like truncation modalities or cohesion (shape and sharp) modalities \cite{hottbook,shulman2015brouwer,schreiber2014quantum}. A caveat there is that the identity types in the model are intensional and that many equations are often only assumed up to propositional rather than judgemental equality.

\clearpage
\subsection{Operational Semantics and Effects}\label{sec:depop} 
We define an operational semantics on the fragment of dCBPV- without complex values. However, we may still want to allow complex values to occur in types, as this greatly increases the expressivity of our type system. (This is in fact the reason we have chosen to discuss complex values in this paper.) In case $\ct{B}$ is a type that has complex values as a subterm, by contrast with the simply typed situation, it may no longer be true that for any $\Gamma\vdash^c M:\ct{B}$, we have an $\Gamma\vdash^c \widetilde{M}:\ct{B}$ without complex values which is judgementally equal. Therefore, we may want to add the rule of figure \ref{fig:dcbpvextrarule} to dCBPV-.
\begin{figure}[!ht]\fbox{
\parbox{\linewidth}{
\AxiomC{$\Gamma\vdash^c M= \return{V}: FA$}
\AxiomC{$\Gamma,x:A,\Gamma'\vdash^c N:\ct{B} $}
\BinaryInfC{$\Gamma,\Gamma'[V/x]\vdash^c \toin{M}{x}{N}:\ct{B}[V/x]$}
\DisplayProof
}}
\caption{\label{fig:dcbpvextrarule} An extra rule to add to dCBPV- which does not introduce any new terms, up to judgemental equality, but which makes complex values redundant in computations.}
\end{figure}
\quad\\ We note that this rule introduces no new terms to dCBPV-, at least up to judgemental equality, as $\toin{M}{x}{N}=\toin{\return{V}}{x}{N}=\lbi{x}{V}{N}$. This means that it is invisible to the categorical semantics. It allows us to type a wider range of terms for which we can define an operational semantics, however. In fact, it tells us that we can again safely exclude complex values from computations (even at types whose construction involves complex values!) without losing any expressive power.
\begin{theorem}[Redundancy of Complex Values for Computations] \label{thm:complexvaluesred} For any dCBPV- computation $\Gamma\vdash^c M:\ct{B}$, there is a computation $\Gamma\vdash^c \widetilde{M}:\ct{B}$ which does not contain complex values as subformulae, such that $\Gamma\vdash^c M=\widetilde{M}:\ct{B}$. Moreover, both the CBV and CBN translations only produce complex-value-free computations.
\end{theorem}
\begin{proof}[Proof (sketch)]
The rule \ref{fig:dcbpvextrarule} allows us to apply the usual proof of Proposition 14 in \cite{levy2006call} again, where we treat eliminators for $\Id$-types exactly as we do those for e.g. $\times$-types.
\end{proof}
In the dependently typed setting, we only need a very minor modification of the notion of stack to formulate our operational semantics. For this, see figure \ref{fig:depstacks}.
\begin{figure}[!ht]
\fbox{
\parbox{\linewidth}{
\begin{tabular}{ll}
\AxiomC{$\vdash\Gamma;\ct{C}\ctxt$}
\UnaryInfC{$\Gamma;\ct{C}\vdash^k \nil : \ct{C}$}
\DisplayProof
& 
\AxiomC{$\Gamma,x:A\vdash^c M:\ct{B}$}
\AxiomC{$\Gamma;\ct{B}\vdash^k K:\ct{C}$}
\BinaryInfC{$\Gamma;FA\vdash^k \toin{[\cdot]}{x}{M}::K:\ct{C}$}
\DisplayProof
\\
&\\
\AxiomC{$\Gamma;\ct{B_j}\vdash^k K:\ct{C}$}
\UnaryInfC{$\Gamma;\Pi_{1\leq i\leq n}\ct{B_i}\vdash^k j::K:\ct{C}$}
\DisplayProof\hspace{80pt}
&
\AxiomC{$\Gamma\vdash^v V:A$}
\AxiomC{$\Gamma;\ct{B}[V/x]\vdash^k K:\ct{C}$}
\BinaryInfC{$\Gamma;\Pi_{x:A} \ct{B}\vdash^k V::K:\ct{C}$}
\DisplayProof
\end{tabular}
}
}
\caption{\label{fig:depstacks} The rules for forming dependently typed (simple) stacks.}
\end{figure}
\quad\\
We define a configuration to be a pair $M,K$ of a dCBPV- computation $\Gamma\vdash^c M:\ct{B}$ not involving any complex values and a dependently typed (simple) stack $\Gamma;\ct{B}\vdash^k K:\ct{C}$. The CK-machine that evaluates our computations is again just that of figure \ref{fig:ckmachine} where we add the extra transition and terminal configuration of figure \ref{fig:ckextra}.
\begin{figure}[!ht]
\fbox{
\parbox{\linewidth}{
\textbf{Transition}\\
\begin{tabular}{lllllll}
 $  \idpm{(\refl V)}{x}{}{M}      $ &,& $K$\hspace{40pt}& $\leadsto$ \hspace{40pt}& $  M[V/x]        $ &,& $K$ \\
\end{tabular}\\
\\
\textbf{Terminal Configuration}\\
\begin{tabular}{lll}
$\idpm{z}{x}{M}$ &,& $K$
\end{tabular}
}
}
\caption{\label{fig:ckextra} The additional transition and terminal configuration that specify the operational behaviour of identity witnesses}
\end{figure}\quad\\
As before, we can add the effects of figure \ref{fig:effects} together with their operational semantics of figures \ref{fig:opsemdivs} and \ref{fig:opsemprint} and equations of figure \ref{fig:effeqn}. We get the same determinism, strong normalization and subject reduction results as in the simply typed case.

\begin{theorem}[Determinism, Strong Normalization and Subject Reduction] Every transition respects the type of the configuration. No transition occurs precisely if we are in a terminal configuration. In absence of erratic choice, at most one transition applies to each configuration. In absence of divergence and recursion, every configuration reduces to a terminal configuration in a finite number of steps.
\end{theorem}
\begin{proof}[Proof (sketch)] The proof is in fact no different from that in the simply typed case, as all reductions are defined on untyped configurations and are easily seen to preserve types. All there is to subject reduction it is to observe that each transition preserves the type $\ct{C}$ of the configuration. We note that this result is substantially different from the subject reduction result of \cite{martin1998intuitionistic}, which is much more delicate to prove and relies heavily on the conversion rules of type theory. The reason that we avoid such delicate issues is that types only depend on values (and not computations) which never get reduced in our transitions.\end{proof}

\begin{remark}[Type Checking]
While the operational semantics discussed here is very relevant as it describes the execution of a program of dCBPV-, one could argue that a type checker is as important an operational aspect to the implementation of a dependent type theory. We leave the description of a type checking algorithm to future work. We note that the core step in the implementation of a type checker is a normalization algorithm for directed versions (from left to right) of the equations for values of figures \ref{fig:vceqs}, \ref{fig:vcdepeqs} and \ref{fig:effeqn} (with congruence laws), as this would give us a normalization procedure for types. We hope to be able to construct such an algorithm using normalization by evaluation by combining the techniques of \cite{abel2007normalization} and \cite{ahman2013normalization}. Our hope is that this will lead to a proof of decidable type checking of the system at least in absence of the $\eta$-law for $\Id$-types and without recursion. We note that the complexity of a type checking algorithm can vary widely depending on which equations we include for effects. The idea is that one only includes a basic set of program equations as judgemental equalities to be able to decide type checking and to postulate other equations as propositional equalities, which can be used for manual or tactic-assisted reasoning about effectful programs.
\end{remark}
\clearpage

\section{Dependently Typed Call-by-Push-Value with Dependent Kleisli Extensions}\label{sec:depcbpvklei}
While the system dCBPV- is very clean in its syntax, operational semantics, categorical semantics and admits plenty of concrete models, it may be a bit of a disappointment to the reader who was expecting to see a proper combination of effects and dependent types, rather than a system that keeps both features side by side without them interacting meaningfully\footnote{As we shall discuss later, this is not entirely fair on dCBPV-, as it does allow us to form types (predicates) depending on thunks of effectful computations.}. In particular, one might find it unsatisfactory that the CBV-translation from dependent type theory into dCBPV- fails and that the CBN-translation only goes through to a limited extent.

To address these issues, we shall introduce a more expressive system in this section which we call dCBPV+ and which extends dCBPV- with Kleisli extensions for dependent functions. We leave the discussion about the pros and cons of dCBPV+ compared to dCBPV- until section \ref{sec:concl} and to future work.

\subsection{Syntax}
We have seen the need to add dependent Kleisli extensions in the form of the rule shown in figure \ref{fig:depklext} if we want to obtain a dependently typed equivalent of the CBV-translation into CBPV or if we want to model dependent elimination rules for the positive connectives in the CBN-translation. We shall use the name dCBPV+ to explicitly refer to the resulting system of the rules of dCBPV- (figures \ref{fig:ctxtrules}, \ref{fig:depcbpvtypes}, \ref{fig:vcdepterms}, \ref{fig:dcbpvextrarule}, \ref{fig:vceqs} and \ref{fig:vcdepeqs}) and dependent Kleisli extensions (figure \ref{fig:depklext}).
\begin{figure}[!ht]\fbox{\parbox{\linewidth}{
\AxiomC{$\Gamma,z:UF A,\Gamma'\vdash \ct{B}\ctype$}
\AxiomC{$\Gamma\vdash^c M:FA$}
\AxiomC{$\Gamma,x:A,\Gamma'[\tr x/z]\vdash^c N:\ct{B}[\tr x/z]$}
\TrinaryInfC{$\Gamma,\Gamma'[\thunk M/z]\vdash^c \toin{M}{x}{N} : \ct{B}[\thunk M/z]$.}
\DisplayProof
}}
\caption{\label{fig:depklext} The rule for dependent Kleisli extensions in dCBPV. As before, we write $\tr$ as an abbreviation for $\thunk \return$.}
\end{figure}

We note that, in presence of this extra rule, the translations of figures \ref{fig:depcbvtrans} and \ref{fig:cbntrans} are finally well-defined. We would like to highlight the fact that a type $x_1:A_1,\ldots,x_n:A_n\vdash A\type$ gets translated into a type $z_1:UFA_1,\ldots, z_n:UFA_n\vdash A^v \vtype$ by the CBV-translation. Briefly, this is necessitated by the CBV-translation of substitution of terms in types. For example, to substitute a term $x:B\vdash M:A$ into $x:A\vdash C\type$ in the CBV-translation, we have to be able to substitute $(x:B)^v\vdash^c M^v:FA$ (or equivalently $(x:B)^v\vdash^v \thunk M^v:UFA$) into $(x:A)^v\vdash C^v\vtype$. This forces us to define $(x:A)^v$ as $z:UFA$ if we are to use the usual type substitution of CBPV (after taking the Kleisli extension of $\thunk M^v$).

We would like to say that the CBV- and CBN-translations are sound and complete. However, as no notion of a CBV- or CBN-equational theory has been formulated for dependent type theory, as far as we are aware, we shall take the equational theories induced by these translations as their definitions. Unsurprisingly, $\Sigma$-types behave equationally exactly like $\times$-types and $\Pi$-types do as $\Rightarrow$ types. The interesting connective to study is the $\Id$-type. 
\begin{theorem}
Figures \ref{fig:depcbvtrans} and \ref{fig:depcbntrans} define CBV- and CBN-translations of dependent type theory with $\Sigma_{1\leq i\leq n}$-, $1$-, $\Sigma$-, $\Id$-, $\Pi_{1\leq i\leq n}$- and $\Pi$-types (with dependent elimination rules for all positive connectives) into dCBPV+. In fact, they allow us to transfer an arbitrary theory in CBV- or CBN-dependent type theory to one on dCBPV+ such that we again get well-defined CBV- and CBN-translations. As expected, CBN-$\Id$-types (even extensional ones) satisfy the $\beta$-law but may not satisfy the $\eta$-law. More surprisingly, perhaps, the same is true for CBV $\Id$-types. 
\end{theorem}
\begin{proof}[Proof (sketch)] In the previous section, we have already seen the statement about CBN-$\Id$-types in case we use a non-dependent elimination rule. The case with a dependent elimination rules works similarly.

The interesting case here are the $\Id$-types in the CBV-translation. For the $\beta$-rule, note that
\begin{align*}
(\idpm{\refl M}{x}{N})^v&=\toin{(\toin{(M^v)}{z}{\return\refl\tr z})}{\zeta}{\idpm{\zeta}{y}{\toin{(\force y)}{x}{N^v}}}\\
&=\toin{M^v}{z}{\idpm{\refl\tr z}{y}{\toin{(\force y)}{x}{N^v}}}\\
&=\toin{M^v}{z}{\toin{(\force \thunk \return z)}{x}{N^v}}\\
&=\toin{M^v}{z}{\toin{(\return z)}{x}{N^v}}\\
&=\toin{M^v}{x}{N^v}\\
&=(\lbi{x}{M}{N})^v.
\end{align*}
To see that the $\eta$-rule may fail, consider dCBPV+ with divergence. We note that in case the $\eta$-law would hold for $\Id$-types in CBV-type theory, it would imply the following principle of reflection \cite{jacobs1999categorical}:\\\\
\AxiomC{$\Gamma\vdash P:\Id_A(M,N)$}
\UnaryInfC{$\Gamma\vdash M=N:A.$}
\DisplayProof\\
\\
In particular, presence of divergence would make CBV-type theory identify all terms in that case. In particular, this would mean that the terms $x:A, y:A'\vdash^c \return x:FA$ and $x:A,y:A'\vdash^c\return y:FA'$ are judgementally equal in dCBPV+ with divergence, which they clearly are not. For a formal proof that they are not, we note that we can see this in the families model given in section \ref{sec:dcbpvplusmod}. For a more syntactic intuition, we note that $\Id$-$\eta$ is less harmful in CBPV with effects than it is in CBV or CBN with effects due to the strict distinction between values and computations, as the obvious reflection rule it implies is the following which does not identify all terms in presence of divergence, as it does not trivially let us satisfy the hypothesis of the rule, in the way it did in CBV-type theory.
\\
\\
\AxiomC{$\Gamma\vdash^v V:\Id_A(V_1,V_2)$}
\UnaryInfC{$\Gamma\vdash^v V_1= V_2:A.$}
\DisplayProof
\end{proof}

\clearpage
\subsection{Categorical Semantics}
To formulate the categorical semantics of dCBPV+, we need a dependently typed generalization of the notion of Kleisli triple. A similar notion of dependently typed Kleisli extension has been proposed before in \cite{hottbook} (section 7.7), be it for a more limited class of modalities. In practice, we shall see that for a given indexed adjunction dependent Kleisli extensions may not exist nor be unique.

\begin{definition}[dCBPV+ Model] By a dCBPV+ model, we shall mean a dCBPV- model $F\dashv U:\Ccat\leftrightarrows \Dcat$ equipped with \emph{dependent Kleisli extensions}. That is, maps $$\Ccat(\Gamma.A.\Gamma'\{\proj{\Gamma}{\eta_A }\})(1,UB\{\qu{\proj{\Gamma}{\eta_A}}{\Gamma'}\})\ra{(-)^*}\Ccat(\Gamma.UFA.\Gamma')(1,UB),$$ where $\eta$ is the unit of the adjunction $F\dashv U$, such that the following laws hold for members of the same homset:
\begin{itemize}
\item unitality: $b^*\{\qu{\proj{\Gamma}{\eta_A}}{\Gamma'}\}=b$;
\item composition: $b^*\{\qu{\langle \id_\Gamma,a^*\rangle}{\Gamma'} \}=(b^*\{\qu{\langle \id_\Gamma,a\rangle}{\Gamma'} \})^*$;
\item agreement with the usual non-dependent Kleisli extension $(-)^\star$ for the adjunction $F\dashv U$:
\begin{diagram}
\Ccat(\Gamma)(A,UB) & \rTo^{\cong}& \Ccat(\Gamma.A)(1,UB\{\proj{\Gamma}{A}\})&\;=\;& \Ccat(\Gamma.A)(1,UB\{\proj{\Gamma}{UFA}\}\{\proj{\Gamma}{\eta_A}\})\\ 
\dTo^{(-)^\star} & && &\dTo^{(-)^*} \\
\Ccat(\Gamma)(UFA,UB)&&\rTo^{\cong}&&\Ccat(\Gamma.UFA)(1,UB\{\proj{\Gamma}{UFA}\}).
\end{diagram}
\end{itemize}
\end{definition}
\begin{remark}
Note that it is enough to just specify the dependent Kleisli extensions of the form
$$\Ccat(\Gamma.A.\Gamma'\{\proj{\Gamma}{\eta_A}\})(1,UFA'\{\qu{\proj{\Gamma}{\eta_A}}{\Gamma'}\})\ra{(-)^*}\Ccat(\Gamma.UFA.\Gamma')(1,UFA').$$
Then, we can define, more generally, $f^*:=\lambda_{x:\Gamma.UFA.\Gamma'}(f;\eta_{UB})^*\{x\}; U\epsilon_{B(x)}$, where $\eta$ is the counit of the adjunction $F\dashv U$.
\end{remark}
\begin{remark}[Dependent Costrength]
Note that dependent Kleisli extensions allow us, in particular, to define the \emph{dependent costrength}  $s_{A,B}'$ for the monad $T:=UF$ that we were missing (starting from the identity on $F\Sigma_A (B\{\proj{\Gamma}{\eta_A}\})$):
\begin{align*}
&\quad\Dcat(\Gamma)(F\Sigma_A (B\{\proj{\Gamma}{\eta_A} \}),F\Sigma_A (B\{\proj{\Gamma}{\eta_A} \}))\\
&\cong \Ccat(\Gamma)(\Sigma_A B\{\proj{\Gamma}{\eta_A} \},T \Sigma_A (B\{\proj{\Gamma}{\eta_A} \}))\\
&\cong \Ccat(\Gamma.A.B\{\proj{\Gamma}{\eta_A} \})(1,T\Sigma_A (B\{\proj{\Gamma}{\eta_A} \})\{\proj{\Gamma}{A}\}\{\proj{\Gamma.A}{B\{\proj{\Gamma}{\eta_A} \}}\})\\
&\cong \Ccat(\Gamma.A.B\{\proj{\Gamma}{\eta_A} \})(1,T\Sigma_A (B\{\proj{\Gamma}{\eta_A} \})\{\proj{\Gamma}{TA}\}\{\proj{\Gamma.TA}{B}\}\{\qu{\proj{\Gamma}{\eta_A}}{B}\})\\
&\ra{(-)^*} \Ccat(\Gamma.TA.B)(1,T\Sigma_A (B\{\proj{\Gamma}{\eta_A} \})\{\proj{\Gamma}{TA}\}\{\proj{\Gamma.TA}{B}\})\\
&\cong \Ccat(\Gamma)(\Sigma_{TA}B,T\Sigma_A (B\{\proj{\Gamma}{\eta_A}\})).
\end{align*}
As a consequence, we are able to define both a left and a right pairing (which will in general not coincide):
\begin{diagram}
\Sigma_{TA} {TB} & \rTo^{s'} & T\Sigma_{A} TB\{\proj{\Gamma}{\eta_A}\} &\rTo^{Ts} & T^2 \Sigma_AB\{\proj{\Gamma}{\eta_A}\} & \rTo^{\mu} & T\Sigma_AB\{\proj{\Gamma}{\eta_A}\}\\
\Sigma_{TA} {TB} & \rTo^{s} & T\Sigma_{TA} B&\rTo^{Ts'} & T^2 \Sigma_AB\{\proj{\Gamma}{\eta_A}\} & \rTo^{\mu} & T\Sigma_AB\{\proj{\Gamma}{\eta_A}\}.
\end{diagram}
\end{remark}
\begin{theorem}[dCBPV+ Semantics] We have a sound interpretation of dCBPV+ in a dCBPV+ model. The interpretation in such categories is complete in the sense that an equality of values or computations holds in all interpretations iff it is provable in the syntax of dCBPV+. In fact, if we add the obvious admissible weakening and exchange rules to dCBPV+, the interpretation defines an onto correspondence between categorical models and syntactic theories in dCBPV+ which satisfy mutual soundness and completeness results. This correspondence becomes 1-1 and we obtain completeness for the stack judgement if we include complex stacks.
\end{theorem}

\begin{theorem}[Dependent CBN-Semantics {2}] The (semantic equivalent of the) CBN-translation of DTT with $\Sigma_{1\leq i\leq n}$-, $1$-, $\Sigma$-, $\Id$-, $\Pi_{1\leq i\leq n}$-, $\Pi$-types, where we use the strong (dependent) elimination rules for all positive connectives, into dCBPV+, lets us construct a categorical model of CBN-dependent type theory with the connectives above out of any model of dCBPV+ by taking the coKleisli category for $!=FU$. The interpretation of CBN-dependent type theory is sound and complete for the equational theory induced from dCBPV+:
\begin{align*}
\sem{\ct{B_1},\cdots,\ct{B_n}\vdash \ct{B}}&=\Dcat(U\sem{\ct{B_1}}.\cdots.U\sem{\ct{B_n}})(F1,\sem{\ct{B}})\cong\Dcat_{!}(U\sem{\ct{B_1}}.\cdots.U\sem{\ct{B_n}})(\top,\sem{\ct{B}}).
\end{align*}
\end{theorem}
\begin{theorem}[Dependent CBV-Semantics] The (semantic equivalent of the) CBV-translation of DTT with $\Sigma_{1\leq i\leq n}$-, $1$-, $\Sigma$-, $\Id$-, $\Pi_{1\leq i\leq n}$-, $\Pi$-types, where we use the strong (dependent) elimination rules for all positive connectives, into dCBPV+, lets us construct a categorical model of CBV-dependent type theory with the connectives above out of any model of dCBPV+ by taking the Kleisli category for $T=UF$. The interpretation of CBN-dependent type theory is sound and complete for the equational theory induced from dCBPV+:
\begin{align*}
\sem{A_1,\cdots,A_n\vdash A}&=\Dcat(\sem{A_1}. \cdots.\sem{A_n}\{\eta_{\sem{A_1},\ldots,\sem{A_{n}}}\})(F1,F\sem{A}\{\eta_{\sem{A_1},\ldots,\sem{A_{n}}} \})\\
&\cong\Ccat_T(\sem{A_1}. \cdots.\sem{A_n}\{\eta_{\sem{A_1},\ldots,\sem{A_{n}}}\})(1,\sem{A}\{\eta_{\sem{A_1},\ldots,\sem{A_{n}}} \}).
\end{align*}
Here, $\eta_{\sem{A_1},\ldots,\sem{A_{n}}}$ is inductively defined by $$\eta_{\sem{A_1},\ldots,\sem{A_{k}}}:=\qu{\eta_{\sem{A_1},\ldots,\sem{A_{k-1}}}}{\sem{A_k}};\proj{\sem{UFA_1}.\cdots.\sem{UFA_{k-1}}}{\eta_{\sem{A_k}}}.$$
\end{theorem}

\begin{remark}
We have finally arrived at a notion of a model for CBV-dependent type theory. It seems much less straightforward than the corresponding notion of a model for CBN-dependent type theory as a particular kind of model of pure dependent type theory in which the $\eta$-laws for positive connectives may fail. Then again, a similar phenomenon is already seen in the simply typed case.
\end{remark}

\clearpage
\subsection{Some Basic Models and Non-Models}\label{sec:dcbpvplusmod}
We had already noted that the identity adjunction on any model of pure DTT gives a model of dCBPV+, which demonstrates consistency. However, as we shall see, it is not the case that any model of dCBPV- extends to a model of dCBPV+. In particular, not every indexed monad on a model of pure DTT admits dependent Kleisli extensions. As it turns out, the existence of dependent Kleisli extensions needs to be assessed on a case-by-case basis. Moreover, one indexed monad might be given dependent Kleisli extensions in several inequivalent ways. Therefore, we treat some dCBPV- models for common effects and discuss the (im)possibility of dependent Kleisli extensions.

\subsubsection{A Non-Model: Writing}
We let $\Bcat$ be $\mathsf{Set}$ and $\Ccat$ be $\Fam(\mathsf{Set})$. Let $M$ be a non-trivial monoid, for instance a monoid of strings of ASCII characters. Then, we let $\Dcat$ be the Eilenberg-Moore category for the indexed monad $-\times M$. Now, we note that dependent Kleisli extensions do not have a sound interpretation in this model of dCBPV-. Indeed, it would amount to giving a map
\begin{diagram}
\Fam(\mathsf{Set})(\Gamma.A)(1,B\{\langle \id_\Gamma,\id_A,1_M\rangle \}\times M)&\rTo^{(-)^*}&\Fam(\mathsf{Set})(\Gamma.(A\times M))(1,B\times M)\\
\Pi_{\langle c,a\rangle \in \Gamma.A}B(c,a,1_M)\times M&\rTo^{(-)^*}&\Pi_{\langle c,a,m\rangle \in \Gamma.(A\times M)}B(c,a,m)\times M\\
f=\langle f_B,f_M\rangle & \rMapsto &\lambda_{c,a,m} \langle ?,f_M(c,a)*m\rangle .
\end{diagram}
We see that this is not always possible. For instance, let $\Gamma=1=A$ and let $B$ be a predicate that expresses that $m=1_M$. In that case, any $f^*$ cannot be a total function as it cannot send, for instance, $(*,*,\textnormal{\texttt{hello world}})$ anywhere.\\
\\
We would like to stress that this does not show that dependent Kleisli extensions are incompatible with printing. Indeed, it only shows that this particular model of printing does not admit dependent Kleisli extensions. One could conceive, for example, of a model of printing where types depending on $TA$ are not allowed to refer to what is being printed, in which case we could define $f^*(c,a,m):=\langle f_B(c,a),f_M(c,a)*m\rangle $. More generally, such a definition could work if, for all $m\in M$, $B(c,a,1_M)\subseteq B(c,a,m)$. Importantly, such an indexed category would not be of the shape $\Fam(\Ccat)$, but it would really involve non-trivial type dependency. (Recall that $\Fam(\Ccat)$ is the cofree indexed category over $\Set$ on a category $\Ccat$ \cite{vakar2015syntax}, hence is really a model of simple type theory in disguise as a model of dependent type theory.)

\subsubsection{A Non-Model: Reading}
We let $\Bcat$ be $\mathsf{Set}$ and $\Ccat$ be $\Fam(\mathsf{Set})$. Let $S$ be some non-trivial set (that is, not $0$ or $1$), which we think of as a set of states for our storage cell. Then, we let $\Dcat$ be the Eilenberg-Moore category for the indexed monad $(-)^S$. Now, we note that dependent Kleisli extensions do not have a sound interpretation in this model of dCBPV-. Indeed, it would amount to giving a map
\begin{diagram}
\Fam(\Set)(\Gamma.A)(1,B\{\lambda_s\langle \id_\Gamma,\id_A\rangle \}^S)&\rTo^{(-)^*}&\Fam(\Set)(\Gamma.(A^S))(1,B^S)\\
\Pi_{\langle c,a\rangle \in \Gamma.A}B(s\mapsto \langle c,a\rangle)^S&\rTo^{(-)^*}&\Pi_{(s\mapsto\langle c,a_s\rangle )\in \Gamma.(A^S)}B(s\mapsto \langle c,a_s\rangle)^S\\
f & \rMapsto &\lambda_{s\mapsto \langle c,a_s\rangle}\lambda_{s'} ?.
\end{diagram}
We see that this is not always possible. For instance, let $\Gamma=1$ and $A=2$ and let $B$ be a predicate that expresses that $s\mapsto \langle *,a_s\rangle  $ is constant. In that case, any $f^*$ cannot be a total function as it cannot send a non-constant $s\mapsto \langle *,a_s\rangle$ anywhere.\\
\\
Again, this does not show that reading as an effect is incompatible with dependent Kleisli extensions, per se. One could imagine a model in which, for all fixed $s'\in S$, $B(s\mapsto \langle c,a_{s'}\rangle)\subseteq B(s\mapsto \langle c,a_s\rangle )$. In that case, we can define $f^*(s\mapsto \langle c,a_s\rangle )(s'):=f(c,a_{s'})(s')$.

\subsubsection{A Non-Model: Global State}
Similarly, for global state, we let $\Bcat$ be $\Set$ and $\Ccat$ be $\Fam(\Set)$ and we take $T:= (-\times S)^S$, where $S$ is a non-trivial set, and let $\Dcat$ be the Eilenberg-Moore category for $T$. Then, dependent Kleisli extensions would amount to a map
\begin{diagram}
\Fam(\Set)(\Gamma.A)(1,(B\{\lambda_s\langle \id_\Gamma,\id_A\rangle \}\times S)^S)&\rTo^{(-)^*}&\Fam(\Set)(\Gamma.((A\times S)^S))(1,(B\times S)^S)\\
\Pi_{\langle c,a\rangle \in \Gamma.A}(B(s\mapsto \langle c,a,s\rangle)\times S)^S&\rTo^{(-)^*}&\Pi_{(s\mapsto\langle c,a_s,s_s\rangle )\in \Gamma.((A\times S)^S)}(B(s\mapsto \langle c,a_s,s_s\rangle)\times S)^S\\
f & \rMapsto &\lambda_{s\mapsto \langle c,a_s,s_s\rangle}\lambda_{s'} ?.
\end{diagram}
Now, $B$ could express the property that $a_s=a$ ($a_s$ is independent of $s$) and $s_s=s$. In that case, no such dependent Kleisli extension exists.

One could imagine a different model of global state, however, in which, for every fixed $s'\in S$, $B(s\mapsto \langle c,a_{s'},s\rangle )\subseteq B(s\mapsto \langle c,a_s,s_s\rangle)$. In that case, one could define $f^*(s\mapsto \langle c,a_s,s_s\rangle )(s'):=f(c,a_{s'})(s')$.

\subsubsection{A Model: Exceptions or Divergence}
We consider a model for exceptions or divergence, where we use the either monad $T=E+(-)$  on $\Fam(\Set)$, for some fixed set $E$ whose elements we think of as exceptions or, perhaps, in the case of $E=1$, divergence. We let $\Bcat$ be $\Set$ and $\Ccat$ be $\Fam(\Set)$ and we take for $\Dcat$ the Eilenberg-Moore category for $T$. In this case, we in fact have maps
\begin{diagram}
\Fam(\Set)(\Gamma.A)(1,E+B\{\langle \id_\Gamma,\inr\rangle\})&\rTo^{(-)^*}&\Fam(\Set)(\Gamma.(E+A))(1,E+B)\\
\Pi_{\langle c,a\rangle \in \Gamma.A}E+B(c,\inr\;a)&\rTo^{(-)^*}&\Pi_{\langle c,t\rangle \in \Gamma.(E+A)}E+B(c,t)\\
f & \rMapsto & \lambda_{c}[\inl ,f(c,-)].
\end{diagram}
These are easily seen to give a sound interpretation of dependent Kleisli extensions. They indeed model the propagation of exceptions one would expect.

\subsubsection{A Dubious Model: Erratic Choice}
We consider a model for erratic choice, where we use the powerset monad $T=\p$  on $\Fam(\Set)$. We let $\Bcat$ be $\Set$ and $\Ccat$ be $\Fam(\Set)$ and we take for $\Dcat$ the Eilenberg-Moore category for $T$. Dependent Kleisli extensions would amount to maps
\begin{diagram}
\Fam(\Set)(\Gamma.A)(1,\p B\{\langle \id_\Gamma,x\mapsto \{x\} \rangle\})&\rTo^{(-)^*}&\Fam(\Set)(\Gamma.(\p A))(1,\p B)\\
\Pi_{\langle c,a\rangle \in \Gamma.A}\p B(c,\{a\})&\rTo^{(-)^*}&\Pi_{\langle c,t\rangle \in \Gamma.(\p A)}\p B(c,t)\\
f & \rMapsto & \lambda_{c,t}?.
\end{diagram}
We can, in principle, define $f^*(c,t):=(\bigcup_{a\in t} f(c,a))\cap B(c,t)$ to obtain a dependent Kleisli extension. However, this model might not correspond to the expected operational semantics. It would be preferable to consider, instead, a model of type theory $\Ccat$ in which it is always the case that $\bigcup_{a\in t} B(c,\{a\})\subseteq  B(c,t)$, which case we can just define $f^*(c,t):=\bigcup_{a\in t}f(c,a)$ (cf. reader monad).

\subsubsection{A Puzzle: Control Operators}
We consider a dCBPV- model for control operators, where we use a continuation monad $T=R^{(R^-)}$  on $\Fam(\Set)$, for some non-trivial set $R$. We let $\Bcat$ be $\Set$ and $\Ccat$ be $\Fam(\Set)$ and we take for $\Dcat$ the Eilenberg-Moore category for $T$. Dependent Kleisli extensions would amount to maps\pagebreak
\begin{diagram}
\Fam(\Set)(\Gamma.A)(1, (R^{(R^{B\{\langle \id_\Gamma,x\mapsto \mathsf{ev}_x \rangle\}})}) &\rTo^{(-)^*}&\Fam(\Set)(\Gamma.( R^{(R^A)}))(1,R^{(R^B)})\\
\Pi_{\langle c,a\rangle \in \Gamma.A}(R^{(R^{B(c,\mathsf{ev}_a)})})&\rTo^{(-)^*}&\Pi_{\langle c,t\rangle \in \Gamma.(R^{(R^A)})}(R^{(R^{B(c,t)})})\\
f & \rMapsto & \lambda_{c,t}?.
\end{diagram}
In order to match the expected operational semantics, it is tempting to try and define, just as in the simply typed case, $f^*(c,t)(k):=t(\lambda_a f(c,a)(k))$. However, this is only well-defined if we have $\forall_{a\in A(c)}R^{B(c,t)}\subseteq R^{B(c,\mathsf{ev}_a)}$. In particular, we would have that $B(c,\mathsf{ev}_A)=B(c,\mathsf{ev}_{A'})$ for all $a,a'\in A(c)$. This suggests a kind of incompatibility between control operators and dependent Kleisli extensions. We would like to further investigate the combination of dCBPV with control operators in future work, especially given the correspondence with classical logic. In the light of \cite{herbelin2005degeneracy}, we already know that the combination of dependent types and control operators can easily lead to logical inconsistency of the system (in the sense that every type becomes inhabited). The paradox above, however, would point to an incompatibility of a different kind, more to do with subject reduction rather than type inhabitation. It remains to be seen if the problem can be avoided by using  more sophisticated models of control.

\subsubsection{A Model: Recursion}
Based on the Scott domains (i.e. bounded complete, directed complete, algebraic cpo) model of dependent type theory with recursion of \cite{palmgren1990domain}, we can construct a model of dCBPV- (with intensional $\Id$-types). We take $\Bcat$ to be the category of Scott predomains (i.e. disjoint unions of Scott domains) and continuous functions. For a preorder-enriched category like $\Bcat$, we call a pair of morphisms $e:A\leftrightarrows B:p$ an embedding-projection pair if $e;p=\id_A$ and $p;e\leq \id_B$. We write $\Bcat_{ep}$ for the lluf subcategory of $\Bcat$ of the embedding-projection pairs.  We can make this into a model $\Ccat$ of DTT in the same way as we can for the category of Scott domains: we define $\Ccat(A)$ to be the category of Scott predomains parametrised over $A$, which we define to be continuous functors from $A$ into $\Bcat_{ep}$. This supports $1$-, $\Sigma$-, $\Sigma_{1\leq i\leq n}$- and (intensional) $\Id$-types. Briefly, $1$ is the one-point predomain, $\Sigma$-types are just the set-theoretic $\Sigma$-types equipped with the product order, $\Sigma_{1\leq i\leq n}$-types are given by disjoint unions and $\Id_A(x,y):=\{z\in A\;|\; z\leq x \;\wedge z\leq y\}$ with the induced order from $A$.

We can define $\Dcat(A)$ to be the category of Scott domains parametrised over $A$ (continuous families of Scott domains) with strict continuous (families of) functions as morphisms. This extends to give an indexed category of parametrised Scott domains indexed over Scott predomains. This category supports $\Pi_{1\leq i\leq n}$- (the ordinary set-theoretic product equipped the product order) and $\Pi$-types (the set of dependent functions which are continuous, equipped with the pointwise order). All constructions and proofs are exactly as in \cite{palmgren1990domain} with the only difference that in some cases we have to replace the word domain with predomain. We have an indexed adjunction between $\Ccat$ and $\Dcat$: the inclusion of $\Dcat$ into $\Ccat$ has a left adjoint which adjoins a bottom element. In fact, we have a model of dCBPV-. The model clearly supports recursion, as we can define our usual fixpoint combinators. This model is easily seen further to support dependent Kleisli extensions: similar to our previous model of divergence, for a dependent function $f$, we define the Kleisli extension $f^*$ to send the new bottom element to bottom and otherwise act as $f$. In fact, this model also supports universes for both value and computation types \cite{palmgren1993information} and the dependent projection products and (intensional) negative $\Id$-types of section  \ref{sec:depprojprod}. We note that the dependent projection product of Scott domains is just the $\Sigma$-type of \cite{palmgren1990domain} (which generalizes the categorical product), while the pattern-matching $\Sigma$-type of Scott domains is its lift. Similarly, the (positive) $\Id$-type of Scott domains is a lifted version of the negative $\Id$-type. We note that the $\Sigma$- and $\Id$-types described in \cite{palmgren1990domain} are the negative versions.

\subsubsection{A Hope: Recursion, Control Operators and Local State}
The game semantics for dependent type theory of \cite{abramsky2015games} is easily seen to extend to a situation where we drop the various conditions on strategies that exclude effects. This gives us a model of CBN-dependent type theory with recursion, (a limited form of) control operators and local references of ground type (and hopefully even general references). The hope is that this extends to a model of dCBPV- similar to the game semantics for simple CBPV \cite{levy2012call}. Seeing that the induced monad would simply be lifting, we would expect this to in fact give a model of dCBPV+.

\clearpage
\subsection{Operational Semantics and Effects}
Using the CK-machine, we can again define an operational semantics on the complex value free computations of dCBPV+. We first note that, again, complex values are redundant in the sense of theorem \ref{thm:complexvaluesred}.

The definition of the operational semantics does not change in presence of dependent Kleisli extensions and is exactly as that described in section \ref{sec:depop}. In particular, figures \ref{fig:ckmachine} and \ref{fig:ckextra} define a CK-machine on which we evaluate the (complex value-free) computations of pure dCBPV+, using the stacks of figure \ref{fig:depstacks}. As before, we can add the effects of figure \ref{fig:effects} together with their operational semantics of figures \ref{fig:opsemdivs} and \ref{fig:opsemprint} and equations of figure \ref{fig:effeqn}.  In any case, we still have the same determinacy and strong normalization results as before, as the essentially untyped proofs remain valid.
\begin{theorem}[Determinacy, Strong Normalization] No transition occurs precisely if we are in a terminal configuration. In absence of erratic choice, at most one transition applies to each configuration. In absence of divergence and recursion, every configuration reduces to a terminal configuration in a finite number of steps.
\end{theorem}
However, the results of section \ref{sec:dcbpvplusmod} are reflected at the level of the operational semantics. While for some effects like divergence, exceptions and recursion, subject reduction is entirely unproblematic to establish. We need to add certain subtyping rules to obtain subject reduction in presence of printing, global state and erratic choice.
\begin{theorem}[Subject Reduction 1] In absence of printing, global state and erratic choice, every transition of the CK-machine from a well-typed configuration, results in a well-typed configuration of the same type. \end{theorem}
\begin{proof}[Proof (sketch)] Due to the rule for dependent Kleisli extensions, we have to take care that every transition still results in a well-typed configuration, if we started with one. In that case, it is clear from the definition of the CK-machine that the type of the configuration does not change during the transition. Suppose we start from a configuration $M,K$ with $\Gamma\vdash^c M:\ct{B}$ and $\Gamma;\ct{B}\vdash^k K:\ct{C}$. What could conceivably happen during a transition is that the resulting configuration $M',K'$ can be given types $\Gamma\vdash^c M':\ct{B}'$  and $\Gamma';\ct{B}''\vdash^k K':\ct{C}$, but not such that $\Gamma\vdash \ct{B}'=\ct{B}''$.

The only transition that can possibly cause this problem is the one for initial computation $\return V$ (see figure \ref{fig:ckmachine}). Indeed, what could have happened is that we started from a configuration $\toin{M}{x}{N},K$, with $\Gamma\vdash^c \toin{M}{x}{N}: \ct{B}[\thunk M/z]$, which transitions to $M,\toin{[\cdot]}{x}{N}::K$. After this, we continue reducing $M$, until we end up in a configuration $\return V, \toin{[\cdot]}{x}{N}::K$. It is now critical when we apply one more transition (the dangerous one) that for the resulting configuration $N[V/x],K$ we again have that $\Gamma\vdash N[V/x]:\ct{B}[\thunk M/z]$, while we know that $\Gamma\vdash N[V/x]:\ct{B}[\tr V/z]$.

This indeed turns out to be the case in pure dCBPV+, as all minimal sequences of transitions of the shape $M,K\leadsto M',K$ just consist of an application of a directed version (left to right) of one of the judgemental equalities of figure \ref{fig:vceqs} or \ref{fig:vcdepeqs}. Further, recall that $V= V'$ implies that $\ct{B}[V/x] = \ct{B}[V'/x]$. In particular, we have that $M= \return V$ so $\ct{B}[\thunk M/z]=\ct{B}[\tr V/z]$ in the critical transition above. Here, we know that $N[V/x]:\ct{B}[\tr V/x]$ and, therefore, by the conversion rule, $N[V/x]:\ct{B}[\thunk M/x]$.

The same argument applies without problems when we add errors, divergence or recursion to dCBPV+. Indeed, after an error, no transition applies so we are safe. $\diverge$ only transitions to $\diverge$ so this does not pose any problems either. The transition for the fixpoint combinator does not break subject reduction as we have included it as a judgemental  equality in figure \ref{fig:effeqn}.
\end{proof}
The proof above shows why subject reduction may fail in presence of printing, global state and erratic choice: the effects $M,K\leadsto M',K$ of their transitions of figures \ref{fig:opsemdivs} and \ref{fig:opsemprint} on computations are not contained in the judgemental equalities we consider (see figure \ref{fig:effeqn}) in the sense that not $M=M'$. In this sense, they differ from the other transitions we have considered. In fact, it is not reasonable to demand such an equality. In particular, in the case of reading global state and erratic choice, that would lead to all computations being equated.

On closer inspection, however, all we really needed to establish subject reduction was an inclusion of computation types, whenever $M,K,m,s\leadsto M',K,m',s'$,\\
\\
\AxiomC{$\Gamma\vdash^c N:\ct{B}[\thunk M'/z]$}
\UnaryInfC{$\Gamma\vdash^c N:\ct{B}[\thunk M/z].$}
\DisplayProof
\\\\
These rules may in fact be reasonable to include in our type system. The idea is that the type of a computation becomes more determined in the computation progresses. We list them in figure \ref{fig:dcbpvplussub}.
\begin{figure}[!ht]
\fbox{
\parbox{\linewidth}{
\begin{tabular}{ll}
\AxiomC{$\Gamma\vdash^c N:\ct{B}[\thunk M/z]$}
\UnaryInfC{$\Gamma\vdash^c N:\ct{B}[\thunk (\print{m}{M})/z]$}
\DisplayProof\hspace{60pt}
&
\AxiomC{$\Gamma\vdash^c N:\ct{B}[\thunk M_{i'}/z]$}
\UnaryInfC{$\Gamma\vdash^c N:\ct{B}[\thunk (\nondet{i}{M_i})/z]$}
\DisplayProof\\
&\\
\AxiomC{$\Gamma\vdash^c N:\ct{B}[\thunk M)/z]$}
\UnaryInfC{$\Gamma\vdash^c N:\ct{B}[\thunk (\writecell{s}{M})/z]$}
\DisplayProof
&
\AxiomC{$\Gamma\vdash^c N:\ct{B}[\thunk M_{s'}/z]$}
\UnaryInfC{$\Gamma\vdash^c N:\ct{B}[\thunk (\readcell{s}{M_s})/z]$}
\DisplayProof
\end{tabular}
}
}
\caption{\label{fig:dcbpvplussub} The rules we add to dCBPV+ to restore subject reduction in presence of printing, global state and erratic choice.}
\end{figure}\quad\\
Specifically, we see that a computation type in which we substitute (the thunk of) a computation with multiple branches (erratic choice, reading state) has to contain each of the types obtained by substituting (the thunk of) any of the branches instead. Similarly, a type in which we substitute (the thunk of) a computation which writes (printing, writing a state) and then performs a computation $M$ has to contain the type obtained by just substituting $M$.

As anticipated, these rules restore the subject reduction property for dCBPV+ with printing, global state and erratic choice.
\begin{theorem}[Subject Reduction 2] If we add the rules of figure \ref{fig:dcbpvplussub} to dCBPV+, each transition of the CK-machine from a well-typed configuration, results in a well-typed configuration with the same type, even if we add printing, global state and erratic choice to the calculus.\end{theorem}

We note that the rules of figure \ref{fig:dcbpvplussub} break uniqueness of typing (which dCBPV+ does otherwise have, presuming that we sufficiently annotate $\diverge$, $\error{e}$ and $\mu_z M$ with types). This may not be desirable. We would suggest to formalize the system using a notion of subtyping \cite{aspinall1996subtyping,luo1999coercive} which should allow us to establish the existence of a minimal type (rather than a unique type) for each well-typed term.

\begin{remark}[Type Checking] Type checking has been studied in systems that combine dependent types with subtyping and has been shows to be decidable in simple cases \cite{aspinall1996subtyping}. Seeing that we are working with a relatively simple form of subtyping, we have hope that we would be able to extend a type checking algorithm for dCBPV- to dCBPV+. An important first step, however, would be the development of a type checker for dCBPV-.
\end{remark}

\clearpage
\section{More Connectives}
In this section, we give a brief outline of how we think one could proceed to extend the dCBPV analysis to more connectives like dependent projection products, dependently typed EEC connectives and inductive type families and universes. Many of the details on these topics remain to be worked out. However, we feel we should at least sketch the situation to the reader as the system of connectives discussed up to this point is not expressive enough for most practical dependently typed programming and theorem proving.

\subsection{Dependent Projection Products and More Negative Connectives}\label{sec:depprojprod}
It was somewhat surprising, perhaps, that while dependent pattern matching products arise so naturally in CBPV, dependent projection products seem less natural. The reader should compare this to the status of additive $\Sigma$-types in linear logic, which are  less straightforward than multiplicative $\Sigma$-types. In principle, we could include the system of rules of figure \ref{fig:addsigma} for them in dCBPV to replace $\Pi_{1\leq i\leq n}$-types.\vspace{-5pt}
\begin{figure}
[!ht]
\fbox{
\parbox{\linewidth}{
\begin{tabular}{ll}
\AxiomC{$\vdash \Gamma,z_1:U\ct{B_1},\ldots,z_n:U\ct{B_n}\ctxt$}
\UnaryInfC{$\Gamma\vdash \Pi_{1\leq i\leq n}^{dep} \ct{B_i}\ctype$}
\DisplayProof
&
\AxiomC{$\{\Gamma\vdash^c M_i:\ct{B_i}[\thunk M_1/z_1,\ldots,\thunk M_{i-1}/z_{i-1}]\}_{1\leq i\leq n}$}
\UnaryInfC{$\Gamma \vdash^c \lambda_i M_i : \Pi_{1\leq i\leq n}^{dep}\ct{B_i}$}
\DisplayProof\\
&\\
&
\AxiomC{$\Gamma\vdash^c M:\Pi_{1\leq i \leq n}^{dep} \ct{B_i}$}
\UnaryInfC{$\Gamma\vdash^c i\textquoteleft M: \ct{B_i}[\thunk 1\textquoteleft M/z_1,\ldots,\thunk (i-1)\textquoteleft M/z_{i-1}]$}
\DisplayProof
\end{tabular}
}
}
\caption{\label{fig:addsigma} Rules for dependent projection products. These should satisfy the obvious $\beta$- and $\eta$-laws.}
\end{figure}\vspace{-5pt}\quad \\
This allows us to define the appropriate CBV- and CBN-translations for dependent projection products in dCBPV, exactly as one defines the translation for simple projection products. This translation re-enforces the idea that the CBV-translation of a type $x_1:A_1,\ldots,x_n:A_n\vdash A\type$ should be $z_1:UFA_1^v,\ldots,z_n:UFA_n^v\vdash A^v\vtype $. We note that we have CBV- and CBN-translations of dependent projection products (which have a dependent/strong elimination principle) even in dCBPV-. Moreover, we can use the usual operational semantics of computations of type $\Pi_{1\leq i\leq n}\ct{B_i}$ for these types.

A curiosity about these type formers is that they do not have a corresponding rule for forming complex stacks, unlike all the other connectives. Hence, although we can formulate a sound and complete categorical semantics for dependent projection products (we demand strong $n$-ary $\Sigma$-types in $\Dcat$ in the sense of objects $\Pi_{1\leq i\leq n}^{dep}\ct{B_i}$ such that $\proj{\Gamma}{U\Pi_{1\leq i\leq n}^{dep}\ct{B_i}}=\proj{\Gamma.U\ct{B_1}.\cdots.U\ct{B_{n-1}}}{U\ct{B_n}};\ldots;\proj{\Gamma}{U\ct{B_1}}$), we cannot include a condition for complex stacks as we have done for the other connectives.

An issue that is perhaps related is the fact that many models in practice fail to support these connectives. In particular, they are hard to obtain in models of linear logic, where they would in particular give additive $\Sigma$-types in the sense of objects $\Sigma_A^\&B$ such that $!\Sigma_A^\&B\cong \Sigma_{!A}^\otimes !B$, and in models of the monadic meta-language, where they would give rise to objects $\Pi_{1\leq i\leq 2}^{dep} A_i$ such that $T\Pi_{1\leq i\leq 2}^{dep} A_i\cong \Sigma_{TA}TB$.

One could add similar negative versions of the other positive connectives like sum types and identity types. Their categorical semantics would correspond to having computation type formers $R(\ct{B_1},\ldots,\ct{B_n})$ that $U$ maps to $R'(U\ct{B_1},\ldots, U\ct{B_n})$ where $R'$ is the corresponding positive type former. One advantage is that one can formulate good CBV- and CBN-translations for these connectives with dependent elimination principles, even in dCBPV-. It remains to be investigated if the obvious operational semantics, where destructors push to the stack and constructors pop the stack and substitute, is well-behaved.

For instance, we could define negative $\Id$-types (which I have called additive $\Id$-types in the case of linear logic \cite{vakar2014syntax}) by the rules of figure \ref{fig:negid}. \hspace{-1pt}Categorically, they correspond to $\Id_{\ct{B}}$ s.t. \mbox{$\hspace{-2pt}\diag{\Gamma}{U\ct{B}}\hspace{-3pt}=\proj{\Gamma.U\ct{B}.U\ct{B}}{U\Id_{\ct{B}}}$.}
\vspace{-10pt}
\begin{figure}
[!ht]
\fbox{\begin{tabular}{ll}
\AxiomC{$\Gamma\vdash^c M_1:\ct{B}$}
\AxiomC{$\Gamma\vdash^c M_2:\ct{B}$}
\BinaryInfC{$\Gamma\vdash \Id_{\ct{B}}(M_1,M_2)\ctype$}
\DisplayProof
&
\AxiomC{$\Gamma\vdash^c M:\ct{B}$}
\UnaryInfC{$\Gamma\vdash^c \mathsf{crefl}\; M : \Id_{\ct{B}}( M, M)$}
\DisplayProof\\
&\\
\AxiomC{$\Gamma\vdash^c M:\Id_{\ct{B}}( M_1, M_2)$}
\AxiomC{$\Gamma,x:U\ct{B}\vdash^c c:C[x/x',\thunk \mathsf{crefl}\; x/p]$}
\BinaryInfC{$\Gamma\vdash^c \mathsf{c}\idpm{M}{x}{c} : C[\thunk M_1/x,\thunk M_2/x',\thunk M/p]$}
\DisplayProof
\end{tabular}
}
\caption{\label{fig:negid} The formation, introduction and elimination rules for negative identity types to which we can add the obvious $\beta$- and $\eta$-rules.}
\end{figure}

\clearpage
\subsection{Dependent EEC and the Relationship with Linear Logic}\label{sec:deec}
As observed by Benton and Wadler \cite{benton1996linear}, linear logic can be seen as the term calculus of stacks for certain commutative effects. The question remained, if more general, possibly non-commutative effects would give rise to a certain kind of generalized, possibly non-commutative linear logic. In particular, the question was if one could define a monoidal-like structure on stacks in a general model of CBPV which generalizes the tensor of linear logic and similarly for the lollipop. A partial positive answer to this was given by the Enriched Effect Calculus (EEC) \cite{egger2009enriching}, telling us that any model of simple CBPV fully and faithfully embeds into a model where we have a binary operation $F(-)\otimes -$ (conventionally, somewhat misleadingly, written $!(-)\otimes -$) which takes a value type and a computation type and produces a computation type and a binary operation $U(-\multimap -)$ (conventionally written $-\multimap -$) which takes two computation types to a value type. Both operations extend to terms in an appropriate way. The notation is chosen to be suggestive as these operations do not generalize the plain linear logic operations $\otimes$ and $\multimap$ but rather the above composite connectives that one can define the LNL-calculus \cite{benton1995mixed}. For instance, we obtain the law that $FA\otimes FB\cong F(A\times B)$, showing us in particular that $F1$ is a right unit for the product operation.

Independently, linear dependent type theory forces a similar operation on us if we wish to extend $\otimes $ to a dependent connective \cite{vakar2014syntax}. Because types are only allowed to depend on cartesian types and not linear ones, the best we can do is a multiplicative $\Sigma$-type $\Sigma_{F(-)}^\otimes-$. Seemingly for two very different reasons, the connective $F(-)\otimes -$ seems to be the more robust notion than $-\otimes -$, if one wants to generalize. We believe this is not a coincidence as the semantics of CBPV already forces various notions of dependent type theory on us.

The question now rises if we can obtain a dependently typed EEC, as both linear dependent types and EEC seem to generalize linear logic in orthogonal but compatible directions. We believe such a calculus would indeed be very natural. We only very briefly note how one could formulate its natural categorical semantics. As we can treat it fibrewise, the connective $U(-\multimap-)$ does not pose much of a problem. Its interpretation just corresponds to having objects $U(\ct B\multimap \ct C)$ such that $\Ccat(\Gamma)(1,U(\ct B\multimap \ct C))\cong\Dcat(\Gamma)(\ct B,\ct C)$. The connective $F(-)\otimes -$, however, generalizes to a dependent connective, which we shall write $\Sigma_{F(-)}^\otimes -$. $\Sigma_{FA}^\otimes -$ is interpreted precisely as a left adjoint to $-\{\proj{\Gamma}{A}\}$ in $\Dcat$, satisfying the left Beck-Chevalley condition. One of its properties is that it satisfies the equation $F(\Sigma_AB)=\Sigma_{FA}^\otimes FB$. In particular, it generalizes $F$ as $A\mapsto \Sigma_{FA}^\otimes F1=FA$.

\clearpage
\subsection{(Co)inductive type families and Universes}
We can easily add strong type forming mechanisms like inductive type families \cite{dybjer1994inductive} and type universes to dCBPV. The question if we can formulate satisfactory CBV- and CBN-translations for the resulting system requires some thought, however. Therefore, for the sake of presentation, we have refrained from adding these mechanisms.

Inductive type families should be included as type formers on value types. Their CBN-translation would apply $U$ to all input types of the construction, apply the inductive type construction and apply $F$ to the result. While the CBV-translation of mere inductive types is straightforward as it can keep the type fixed, one probably has to translate the more complicated inductive families by applying $UF $ to their index types, as we have done for $\Id$-types. This is necessary if we are to substitute CBV terms into the translated types. It remains to be seen if all $\eta$-laws survive this translation as the $\eta$-laws for $\Id$-types seem to have survived in a less robust way than the $\eta$-laws for the other types. Similarly, coinductive types should be treated as computation types. An elegant account of recursive types in CBPV can be found in \cite{levy2012call}. It should be investigated if these can be generalized to include type dependency in the style of Dybjer's inductive families. 

It is not yet clear if one can obtain satisfactory CBV- and CBN-translations for more expressive schemes like induction-recursion, of which type universes (\`a la Tarski) are a special case \cite{dybjer2000general}. As a rule of thumb, inductive-recursive families are more like an inductive than a coinductive construction, hence one would expect them to arise as value types. In the particular case of universes, we would expect separate universes $\mathcal{U}_v$ and $\mathcal{U}_c$ (both of which are value types) to classify value and computation types respectively. It seems most natural to include rules like\\
\\
\AxiomC{$\vdash\Gamma\ctxt$}
\UnaryInfC{$\Gamma\vdash^v 1:\mathcal{U}_v$}
\DisplayProof\\
\\
to build values of the universes which code for types and rules like\\
\\
\AxiomC{}
\UnaryInfC{$x:\mathcal{U}_v\vdash \mathsf{El}_v(x)\vtype$}
\DisplayProof\\
\\
to make types out of codes.

However, this would not be enough to obtain CBV- and CBN-translation for universes, it seems. For that purpose, one would have to replace the latter rules with rules like\\
\\
\AxiomC{}
\UnaryInfC{$x:UF\mathcal{U}_v\vdash \mathsf{El}_v(x)\vtype.$}
\DisplayProof
\\
\\
A notable feature of the resulting system is that the CBV-translation of a code $x_1:A_1, \ldots x_n:A_n\vdash M:\mathcal{U}$ for a type family would be $x_1:A_1, \ldots x_n:A_n\vdash^c M^v :F \mathcal{U}_v$. We can apply a Kleisli extension to this to obtain a code for a value type depending on the context $z_1:UFA_1, \ldots z_n:UFA_n$.

It is not clear that we actually need CBV- and CBN-translations for the whole system. Perhaps it is more useful to work with ordinary universes added to dCBPV on the value side and to avoid effects occurring in types. That is, perhaps the clear separation in CBPV of pure values from effectful computations is a feature that makes universes more tractable than in the usual effectful CBV- and CBN-type theories. Either way, it is clear that a thorough investigation of the issue is warranted. In particular, the enigmatic nature of CBV-type dependency (dependency on $A$ vs $UFA$) deserves to be clarified further. 
\clearpage

\section{dCBPV as a Language for Certified Effectful Programming?}
In this paper, we have discussed two systems which add dependent types to Levy's CBPV: dCBPV- and dCBPV+. The latter extends the former for theoretical reasons with an extra Kleisli extension principle for dependent functions. The reader may wonder what the pros and cons are of both systems, from a practical point of view. One might also wonder what the relationship is between dCBPV and for instance Hoare type theory, which is an existing framework for certified effectful programming based on dependent type theory \cite{nanevski2006polymorphism,nanevski2008ynot}. While a proper investigation of these issues will be the topic of future work, we share some first thoughts on the matter.
 
\subsection{Dependent Kleisli Extensions: a Bug or a Feature?} Let us first say a bit about how we imagine an implementation of either system could be useful in practice. Dependent types are currently used for two closely related\footnote{Both may be said to be instances of certified programming, where the former places the emphasis on the certification and the latter places it on the programming.} but slightly different purposes: theorem proving and functional programming with an expressive type system which can be used to catch a large class of bugs at compile time. For the former purpose, it is important that we disallow computational effects as they may not correspond to sound logical principles under the Curry-Howard correspondence - think for instance of divergence which would give a trivial proof of any proposition\footnote{In fact, even more innocent seeming principles can easily jeopardise the logical consistency (the existence of uninhabited types) of dependent type theory. For instance, the control operator $\mathsf{call/cc}$, which corresponds to the reasonable seeming principle of double negation elimination in the sense of constructive classical logic, degenerates the logical content of dependent type theory with $\Sigma$- and $\Id$-types if it is available at all types \cite{herbelin2005degeneracy}. }. For the latter, it would be very beneficial to have effects at our disposal.

This makes a monadic (or adjunction, in the case of CBPV) encapsulation of effects particularly suited for a dependently typed system, as it allows us to let effects occur in some parts of our code -the parts we want to read as an algorithm - but not other parts - the parts we want to read as a proof. Perhaps, this should be seen as an argument against the practical use of CBV- and CBN-dependent type theories with first class effects which can occur in an unrestricted fashion. Therefore, this might invalidate one of the arguments to prefer dCBPV+ over dCBPV- as we might not be necessary to define CBV- and CBN-translations of effectful dependent type theory in dCBPV. However, it should be emphasized that dependently typed systems with first class recursion are in fact used in for instance Agda and Idris where a termination check to establish logical consistency is performed separately from type checking \cite{norell2007towards,brady2013idris}.\\
\\ 
In a framework for certified effectful programming based on dCBPV, we imagine algorithms are written as (complex value-free) computations which allows us to use various effects in their construction while we use the value judgement for type checking proofs (and we add $\Pi$-constructors for value types). One could also imagine using various different computation judgements, which allow different effects (or even none at all). Importantly, we can form thunks of computations, which allows predicates to refer to algorithms, meaning that we can use the value judgement to prove properties of the algorithms written with the computation judgement.

Recall that most common use cases of the expressive type system of a pure dependently typed language, like Agda, (ignoring features of polymorphism for the moment) rely on types of the form $C = \Sigma_AB$ where $A$ is some type which is expressible in a simply typed language with (co)inductive types, like Haskell, and $B$ is some (perhaps proof relevant) predicate. \cite{bove2009dependent,norell2007towards} Some examples include $C$ a type of vectors of a fixed length, sorted lists, heaps or binary search trees, red-black trees, suitably balanced trees and a type of $\lambda$-terms up to $\beta\eta$-equality. An algorithm representing an inhabitant of such a type can be seen (at least) in one of two ways: either as a single object $c$ or as a pair $\langle a, b\rangle$ of an algorithm which produces an inhabitant $a$ of $A$ and a proof $b$ that property $B$ holds for this inhabitant.

It is the latter point of view which generalizes better to dCBPV. Indeed, we might want to allow effects in the construction of the algorithm a but not in the proof $b$. In other words, we are likely to be interested in (value) inhabitants of types $C = \Sigma_{UFA}B$ where $A$ is a simple Haskell type and $B$ is some predicate, rather than types of the form $UF\Sigma_AB$.\\
\\
As far as implementation goes, it seems plausible that a type checker could be written at least for dCBPV- with as judgemental equality the least congruence relation generated by $\beta\eta$-equality (not including $\eta$ for $\Id$-types) and algebraicity laws for effects (and the fixpoint equation for recursion, if we wish to include it and give up decidability of type checking). Ideally, we would also like to include as many other program equations (e.g. the Plotkin-Power equations for lookup/update algebras \cite{plotkin2002notions}) in judgemental equality to make the type checker handle them in reasoning.

This might not always be possible, however, for instance because we do not want to give up decidability of type checking or simply because it might not be practically feasible to write an efficient type checker which includes all interesting equations. The remaining equalities would have to be postulated as axiomatic values inhabiting certain $\Id$-types (propositional equalities). In proofs, these are the subject of explicit reasoning by the using: they will have to be dealt with either manually or interactively using tactics.
\\
\\
While it may seem that dependent Kleisli extensions are an important feature for the compositionality of the system, it is not, in fact, clear that they are strictly needed for certified effectful programming. Indeed, according to the point of view sketched in the previous paragraph, algorithms are always non-dependent functions, for which we do in fact have a sequencing (or Kleisli extension) principle. It is only proofs of predicates about these algorithms which arise as dependent functions. These, however, should not exhibit effects anyway so there is no need to perform a Kleisli extension on them.

Moreover, it should be noted that dependent Kleisli extensions come at a cost. For instance, we expect type checking to be a significantly more complicated problem in dCBPV+ than in dCBPV-.

Additionally, the subtyping restrictions which are necessary to ensure subject reduction in dCBPV+ might exclude us to express certain useful predicates in our type system that would be unproblematic in dCBPV-. For instance, while it might be possible to express predicates which state that a function prints a specific string (among other things), we cannot express in our type system the property that it does not print a specific string. Similarly, we might be able to express the property that an algorithm exhibits non-determinism, but we cannot express that it has merely two branches.

Furthermore, we are able to obtain suitable CBV- and CBN-translations of effectful dependent type theory into dCBPV-, at least for negative connectives. This means that as long as we are willing to work with negative equivalents of the usual positive connectives in the sense of section \ref{sec:depprojprod}, dependent Kleisli extensions are not necessary for satisfactory  CBV- and CBN-translations.

On the whole, we are currently inclined to view dependent Kleisli extensions as technical devices that were introduced for theoretical reasons, but which may not be very suitable for practical implementations of dCBPV. The extra complexity they introduce into the implementation of a type checker for may not be worth it unless an important new application of dependent Kleisli extensions is found.

\subsection{Comparison with HTT}
Although a proper comparison will be the subject of future work, we would like to make at least a few basic observations on the relationship between dCBPV and an existing successful framework for certified effectful programming which is also based on dependent type theory: Hoare type theory (HTT) \cite{nanevski2006polymorphism} (implemented in Ynot \cite{nanevski2008ynot}).

Regarding the motivation behind both systems, the that should be made is that while HTT seems to have been developed from the start with the syntactic goal in mind of a language for verifying effectful programs, dCBPV arose almost entirely from semantic considerations. In particular, dCBPV was motivated by the study of models of dependent type theory which naturally encode effects, like its domain semantics \cite{palmgren1990domain} and game semantics \cite{abramsky2015games}, the question if linear dependent type theory \cite{vakar2015syntax} can be interpreted as a dependent type theory with commutative effects and the existing categorical semantics of CBPV with almost begs for a dependently typed generalization \cite{levy2005adjunction}.

Regarding their implementation, HTT expresses a property $\phi$ of an  effectful program $V$ of type $A$ by saying that it inhabits a type $T_\phi A$, where $T_\phi$ are monads which are indexed by formulae $\phi$ formed using a(n) (external) separation logic (which can be implemented using the inductive definitions of Coq). dCBPV sticks closer to the Curry-Howard correspondence in its formulation of properties $\phi$ of an effectful program $M$ of type $FA$: they are types $\phi$ depending on the type of thunks $UFA$ and into which we can, in particular, substitute $\thunk M$ and see if we can construct an inhabitant witnessing the truth of $\phi(\thunk M)$.

\clearpage

\section{Future Work}\label{sec:concl}
In this paper, we gave an extension to the realm of dependent types of Levy's CBPV analysis of CBV- and CBN-$\lambda$-calculi. We hope that this one the one hand can shed some light on the theoretical puzzle of how the relationship between effects and dependent types can be understood. On the other hand, we hope it can provide some guidance in the daunting but worthwhile challenge of combining dependent types with effects in a system which is simultaneously practically useful and elegant. To further achieve these two goals, there are several directions in which we would like to take this work.

Firstly, an immediate priority is to describe a type checking algorithm for dCBPV- for a range of effects together with a prototype implementation of the whole system. As discussed, to achieve this, we hope to rely on existing normalization by evaluation methods for pure dependent type theory on the one hand and for simply typed algebraic effects on the other. It remains to be seen which equations for effects can be included as judgemental equalities - hence be handled by the type checker - and which ones should be postulated as propositional equalities (values inhabiting identity types) - hence be dealt with by hand or using tactics.

Secondly, we would like to answer the question what needs to be done to make dCBPV into a practically useful language for certified effectful programming. In particular, is dCBPV- enough for this purpose, or do we have to resort to the more expressive but involved system dCBPV+? Do type forming mechanisms like inductive families and type universes provide enough expressive power? What are the advantages and disadvantages of dCBPV's different way of reasoning about effects compared to Hoare type theory? 

Thirdly, recently, a game semantics was given for dependent type theory in \cite{abramsky2015games}. This is easily seen to give models for dependent type theory with a range of effects like recursion (non-termination), local ground store and control operators. These should be models of a CBN-calculus with dependent projection products (and various other connectives). One of the difficulties in constructing the game semantics for dependent types was that there was no existing body of literature which discusses what exactly qualifies as such an effectful model of dependent type theory. For instance, which equations should we require to hold? We believe the present paper provides a satisfactory answer to these questions. In future work, we hope to demonstrate that this game semantics naturally extends to give a model of dCBPV+ with all connectives we discussed (where the identity types are intensional) which, depending on the variety of strategies we consider, supports recursion, local state and control operators (and perhaps general references).

Fourthly, it should be investigated to which extent the CBV- and CBN-translations into dCBPV+ extend to more expressive type forming mechanisms like inductive families and inductive-recursive definitions like type universes. In particular, we hope this would lead to a better understanding of the rather enigmatic nature of CBV-type dependency.

Finally, it remains to be seen what the status is of subject reduction in dCBPV+ in presence of other effects. We are thinking in particular of control operators, local state and general references. 
\clearpage
\bibliographystyle{splncs}
\bibliography{tau}

\end{document}